\documentclass[journal]{IEEEtranTCOM}

\ifCLASSINFOpdf
\else
\fi
\usepackage{graphicx}
\usepackage{bm}
\usepackage{amsfonts}
\usepackage{amsmath}
\usepackage{amsthm}
\usepackage{amssymb}
\usepackage{enumerate}
\usepackage{times}
\usepackage{subfigure}
\usepackage{latexsym,bm,amssymb} %LatexÖÐÊýѧ°ü
\usepackage{array}
\usepackage{booktabs}
\usepackage{xcolor}
\usepackage[linesnumbered,ruled,vlined]{algorithm2e}
\newtheorem{thm}{Theorem}
\theoremstyle{definition}
\newtheorem{lem}{Lemma}

\newtheorem{myDef}{Definition}
\newtheorem{pro}{Proposition}
\usepackage{balance}
\usepackage{makecell,multirow}
\usepackage{epstopdf}
\begin{document}
\title{Two-Stage Polarization-Based\\Nonbinary Polar Codes for 5G URLLC}

\author{Peiyao~Chen,~Baoming~Bai,~\IEEEmembership{Senior Member,~IEEE},
Xiao~Ma,~\IEEEmembership{Member,~IEEE}% <-this % stops a space
\thanks{
	This work was supported in part by the National Natural Science Foundation of China under Grants 61771364.
	
	P. Chen, and B. Bai are with the State Key Laboratory of Integrated Services Networks, Xidian University, Xi'an 710071, China (E-mail: pychen@stu.xidian.edu.cn, bmbai@mail.xidian.edu.cn).
	
	X. Ma is with the School of Data and Computer Science, Sun Yat-sen University, Guangzhou 510006, China (E-mail: maxiao@mail.sysu.edu.cn).
    }
	%We wish to thank Mr. Huaan Li and Mr. Zhen Liu for providing the performance of LDPC codes. This work was supported in part by the National Natural Science Foundation of China under Grants 61771364.}
}

\maketitle

\begin{abstract}
In this paper, we propose a new class of nonbinary polar codes with two-stage polarization for ultra-reliable and low-latency communications (URLLC), where the outer (\textit{symbol-level}) polarization is achieved by using a $2\times 2$ $q$-ary matrix $\left[ {\begin{smallmatrix}
1&0\\
\beta&1
\end{smallmatrix}}\right]$ as the kernel and the inner (\textit{bit-level}) polarization is achieved by a binary polarization matrix for each input symbol. With the two-stage polarization, \textit{bit-level} code construction is introduced, resulting in \textit{partially-frozen symbols}, where the frozen bits in these symbols can be used as \textit{active-check} bits to facilitate the decoder. The encoder/decoder of the proposed codes has the same structure as the original binary polar codes, admitting an easily configurable and flexible implementation, which is an obvious advantage over the existing nonbinary polar codes based on Reed-Solomon (RS) codes. To support high spectral efficiency in URLLC, we also present, in addition to the \textit{single level coded} modulation scheme with field
matched modulation order, a \textit{mixed multilevel coded} modulation scheme with arbitrary modulation order to trade off the latency against complexity. Simulation results show that our proposed nonbinary polar codes exhibit comparable performance with the RS4-based polar codes and outperform binary polar codes with low decoding latency, suggesting a potential application for 5G URLLC.
\end{abstract}

\begin{IEEEkeywords}
Decoding latency, multiplicative repetition, nonbinary polar codes, two-stage polarization, URLLC. 
\end{IEEEkeywords}
\IEEEpeerreviewmaketitle

\section{Introduction}
Binary polar codes~\cite{Arikan2009}, which have been adopted in 5G for the enhanced mobile broadband (eMBB) control channel, are also considered as promising candidates for ultra-reliable and low-latency communications (URLLC). URLLC requires high reliability as well as low latency even down to millisecond (ms) level for communications. However, the latency of the binary polar codes is relatively high due to the serial processing nature. To reduce the latency, many efforts have been made to improve the degree of parallelism in decoding. In this aspect, nonbinary polar codes provide an effective solution to low decoding latency.

In~\cite{Sasoglu2009}, \c{S}a\c{s}o\v{g}lu \textit{et al.} proved that nonbinary polar codes with arbitrary finite input-alphabet sizes can achieve symmetric-capacity by decomposing underlying symbol channels into a set of subchannels with prime input alphabet sizes. They also showed that all discrete memoryless channels (DMCs) can be polarized by randomized constructions.
In 2014, Chiu~\cite{Chiu2014} proposed a new approach, by using channel symbol permutations, proving that polar codes polarize arbitrary $q$-ary input randomized channels. Nonbinary polar codes based on an $\ell\times \ell$ $q$-ary Reed-Solomon (RS) matrix $\mathbf{G}_{RS}(q, \ell)$ were proposed by Mori and Tanaka~\cite{Mori20101,Mori2014}, whose exponent \footnote{The probability of block error for polar codes under SC decoding is found to be $\tiny{\mathcal{O}\left( {{2^{-{N^\beta }}}} \right)}$ for any $\beta<\gamma$, where $\gamma$ denotes the exponent of the kernel matrix~\cite{Sat2010}.} is $\log(\ell!)/(\ell \log \ell)$ for all $\ell \leq q$ and can be arbitrarily close to 1 as $\ell$ becomes large. In 2016, Cheng~\cite{Cheng2016} \textit{et al.} applied a four-dimensional RS matrix (RS4) over the finite field GF(4) as the kernel and showed that better error-correcting performance can be achieved compared with binary polar codes. Moreover, binary and RS4-based nonbinary kernels are mixed in~\cite{Presman2016} to reduce the decoding complexity of RS4-based polar codes.
Although RS-based nonbinary polar codes exhibit outstanding performance by taking the advantage of large exponents, the very different decoding structures for different field sizes (corresponding to matrix sizes) may limit the applications due to their non-universality. In \cite{Gulcu2018}, the code construction method and its complexity for $q$-ary polar codes over the $q$-ary symmetric channel were introduced. In 2018, we introduced the concept of two-stage polarization to construct nonbinary polar codes in terms of bit-level polarization \cite{Chen2018,Chen20182}. Then, Yuan \textit{et al.}~\cite{Yuan2018} constructed nonbinary polar codes with only single polarization.

The objective of this paper is to construct nonbinary polar codes for URLLC with similar structures to the original polar codes. With this, the conventional successive cancellation list (SCL) decoding~\cite{Vardy2011} and cyclic redundancy check (CRC)-aided SCL decoding methods~\cite{Kai2012} can be easily implemented adaptively. Simulations and analysis results show that, in addition to the improvement in error-correcting performance, the proposed nonbinary polar codes can also achieve a low decoding latency for the reason that multiple bits are decoded simultaneously as a symbol.

The main contributions of this work are summarized as follows:

\begin{itemize}
\item A new class of nonbinary polar codes based on two-stage polarization is proposed. Inspired by the proof of nonbinary polarization in~\cite{Sasoglu2009,Chiu2014} and \cite{Mori2014}, for the symbol-stage polarization, an nonbinary polarized matrix based on a $2\times 2$ $q$-ary matrix $\left[ {\begin{smallmatrix}
1&0\\
\beta&1
\end{smallmatrix}}\right]$ is introduced, where $\beta\neq 0$ acts as a multiplier and varies for nested Kronecker operations.
Then, for the bit-stage polarization, the \textit{bit-level} code construction is considered, which is different from existing constructions and will result in some \textit{partially-frozen} symbols, i.e., some symbols containing both frozen bits and unfrozen bits. To introduce the \textit{bit-level} code construction, a linear transformation is applied to each input symbol before the encoding, leading to \textit{bit-level} polarization in a symbol.

\item According to the encoding, an efficient SC and list decoding are introduced for the proposed nonbinary polar codes, which has a similar structure to that of the Ar{\i}kan's polar codes~\cite{Arikan2009}. The analysis of decoding complexity and decoding latency are also given, which shows that low decoding latency can be obtained by the proposed nonbinary polar codes at the expense of decoding complexity.

\item To improve the error-correcting performance, an \textit{active-check-used} method is proposed by using the frozen bits in \textit{partially-frozen} symbols to facilitate SCL decoding. For further improvement, CRC-aided nonbinary polar codes are considered.

\item For high spectral efficiency, the proposed nonbinary polar codes combining with high-order modulations are also investigated, where two coded modulation schemes are introduced: \textit{single level coded modulation} scheme and \textit{mixed multilevel coded modulation} scheme. The first one is designed for the finite-field matched modulation order, resulting in low latency, while the latter can provide lower complexity by the use of smaller fields for arbitrary modulation orders.
\end{itemize}

The remainder of this paper is organized as follows. In Section II, notations and definitions used in this paper are introduced. In Section III, we present the encoding and decoding algorithms of the two-stage polarization-based nonbinary polar codes.
In Section IV, an \textit{active-check-used} method is proposed to improve performance. Moreover, the proposed nonbinary polar codes concatenated with CRC codes are also considered. The decoding complexity and decoding latency are analyzed in Section V. The proposed nonbinary polar codes combined with high-order modulation are introduced in Section VI. Finally, conclusion is drawn in Section VII.

\section{Preliminaries}
Let $\mathbb{F}_q$ be the finite field with $q=p^m$, where $p$ is a prime number and $m$ is a positive integer greater than unity. Assume that the finite field $\mathbb{F}_q$ is generated by a primitive polynomial $f(x)= f_0 +f_1x+\cdots+f_{m-1}x^{m-1} +x^m \in \mathbb{F}_p[x]$. Let $\alpha$ be a root of $f(x)$, i.e., a so-called primitive element of $\mathbb{F}_q$. Then $\alpha^{-\infty}=0,1,\alpha,\alpha^2,\ldots,\alpha^{q-2}$ form all the elements of $\mathbb{F}_q$. Let $\mathbb{F}^*_q$ denote $\mathbb{F}_q\backslash\{0\}$ and $\mathbb{F}_p(\beta)$ denote the field extension of $\mathbb{F}_p$ generated by the adjunction of $\beta \in \mathbb{F}_q$. Similarly, $\mathbb{F}_p(\mathcal{I})$ and $\mathbb{F}_p(\mathbf{F})$ represents the field extension of $\mathbb{F}_p$ generated by the adjunction of all elements of $\mathcal{I}\subseteq \mathbb{F}_q$ and the matrix $\mathbf{F}$ over $\mathbb{F}_q$, respectively. For the $q$-ary channel polarization, according to~\cite{Mori2014}, we have Theorem 1.
\begin{thm}\label{thm1}
Any $q$-ary input channel is polarized by the $2\times 2$ invertible
matrix $\mathbf{F}$ over $\mathbb{F}_q$ if and only if $\mathbb{F}_p(\mathbf{\overline F})=\mathbb{F}_q$, where $\mathbf{\overline F}$ is one of the standard forms \footnote{Lower triangular matrices with unit diagonal elements equivalent to $\mathbf{F}$ are called standard forms of $\mathbf{F}$.} of $\mathbf{F}$.
\end{thm}

From Theorem 1, we have the following proposition.
\begin{pro}\label{pro1}
\textit{Any $q$-ary input channel can be polarized by the matrix $\mathbf{F'}$ over $\mathbb{F}_q$ with $\mathbf{F'}=\mathbf{\overline F}_a\otimes\mathbf{\overline F}_b$, where $\mathbf{\overline F}_a$ and $\mathbf{\overline F}_b$ are $2\times 2$ invertible matrices over $\mathbb{F}_q$ with $\mathbb{F}_p(\mathbf{\overline F}_a)=\mathbb{F}_q$ and $\mathbb{F}_p(\mathbf{\overline F}_b)=\mathbb{F}_q$, respectively, and ``$\otimes$" represents the Kronecker product.}
%For an $n\times n$ standard form matrix $\mathbf{\overline G}$, if $\mathbb{F}_p(\mathbf{\overline G})=\mathbb{F}_q$, then any $q$-ary channel can be polarized by the matrix $\mathbf{G'}=\gamma\cdot\mathbf{\overline G}$ with $\gamma \in \mathbb{F}^*_q\mathbb{F}_p(\mathbf{\overline F})=\mathbb{F}_q$.
\end{pro}
\begin{proof}
Since $\mathbf{\overline F}_a$ and $\mathbf{\overline F}_b$ are standard forms with  $\mathbb{F}_p(\mathbf{\overline F}_a)=\mathbb{F}_q$ and $\mathbb{F}_p(\mathbf{\overline F}_b)=\mathbb{F}_q$, the generated matrix $\mathbf{F'}$ is also a standard form with $\mathbb{F}_p(\mathbf{F'})=\mathbb{F}_q$. According to Theorem \ref{thm1}, Proposition \ref{pro1} is followed.
%The standard form of $\mathbf{G'}$ is $\mathbf{\overline G}$. Since all diagonal elements of $\mathbf{G'}$ are nonzero elements, thus $\mathbf{G'}$ is an invertible lower triangular matrix. According to Theorem 1, Corollary 1 is followed.
\end{proof}

In this paper, we consider the finite fields of characteristic 2, i.e., $p=2$.
Then, each element $\alpha^i \in $ $\mathbb{F}_q$ can be represented by a binary vector ${\mathbf{b}}(\alpha^i)$=($b_1, b_2,\ldots, b_{m}$) with $\alpha^i=\sum_{j=1}^{m}b_j\alpha^{j-1}$. Similarly, the $m$-bit vector ($b_1, b_2,\ldots, b_{m}$) can be represented by a $q$-ary symbol by the function $g(b_1, b_2,\ldots, b_{m})=\sum_{j=1}^{m}b_j\alpha^{j-1}$, for example, $g(\mathbf{b}(\alpha^i))=\alpha^i$.
Define the companion matrix~\cite{Lidl1994} of $f(x)$ as \[\mathbf{A}=\left[ {\begin{array}{*{20}{c}}
{0}&1&0&\cdots&0\\
{0}&{0}&1&\cdots&0\\
{\vdots}&{\vdots}&\vdots&\ddots&\vdots\\
{0}&{0}&0&\cdots&1\\
{f_0}&{f_1}&f_2&\cdots&f_{m-1}\\
\end{array}}\right],\]
then $\mathbb{F}_q=\{0,\mathbf{A}^i, 0\leq i\leq q-2\}$ with $\alpha^i \leftrightarrow\mathbf{A}^i$. We call $\mathbf{A}^i$ the \textit{matrix representation} of the element $\alpha^i$. Note that the binary
vector representation of $\beta\alpha^i$ ($\beta\in \mathbb{F}_q$), i.e., $\mathbf{b}(\beta\alpha^i)$, is equal to $\mathbf{b}(\beta)\mathbf{A}^i$. In particular, $\mathbf{b}(\alpha^i)$ is exactly $\mathbf{b}(1)\mathbf{A}^i$.

\begin{myDef}(\textit{Equivalent Binary Matrix})
\label{def1}
Given an $n\times n$ matrix $\mathbf{G}_{n}$ over $\mathbb{F}_q$, we define its equivalent $mn\times mn$ binary matrix $\mathbf{\tilde{G}}_{b}$ by replacing each $q$-ary element of $\mathbf{G}_n$ with its \textit{matrix representation}.
\hfill $\blacksquare$
\end{myDef}

\begin{myDef}(\textit{Linear Transformation of $\beta$})
Given an $m\times m$ invertible matrix $\mathbf{H}_m$ over $\mathbb{F}_2$, let $\mathcal{T}_m(\beta)$ with $\beta\in\mathbb{F}_q$ be equal to the $q$-ary element whose binary vector representation is $\mathbf{b}(\beta)\mathbf{H}_m$, i.e., $\mathcal{T}_m(\beta)=g(\mathbf{b}(\beta)\mathbf{H}_m)$.
\hfill $\blacksquare$
\end{myDef}

Since $\mathbf{H}_m$ is an invertible matrix, we have $\mathcal{T}^{-1}_m(\mathcal{T}_m(\beta)=\beta$. In this paper, we assume that a codeword of a nonbinary polar code over $\mathbb{F}_q$ contains $n$ symbols bearing $K$ information bits. Thus, the code rate is $R=K/N$ with the equivalent code length (in bits) $N=mn$. Other notations used in the paper are shown as follows.\\

\textit{\textbf{Notation}}: Let $\mathcal{B}\subseteq \{1,2,\ldots,N\}$ and $\mathcal{B}^{c}=\{1,2,\ldots,N\}\backslash \mathcal{B}$ denote the index set of unfrozen bits and frozen bits of a nonbinary polar code $\mathcal{C}$, respectively. Similarly, let $\mathcal{A}\subseteq \{1,2,\ldots,n\}$ and $\mathcal{A}^{c}=\{1,2,\ldots,n\}\backslash \mathcal{A}$ denote the index set of unfrozen symbols and frozen symbols of $\mathcal{C}$, respectively, where for any $i\in\mathcal{A}^c$, the corresponded index $(i,j)=m(i-1)+j$ with any $1\leq j\leq m$ belongs to the frozen bits index set, i.e., $(i,j) \in \mathcal{B}^c$ for all $i\in\mathcal{A}^c$ and $1\leq j\leq m$. Let $\pi:\{0,1,\ldots,n-1\}\rightarrow\{0,1,\ldots,n-1\}$ be the bit-reversal permutation that maps an index $i$ with binary representation $(b_{1},\ldots,b_{r})$ into the index $\pi(i)$ with binary representation $(b_{r},\ldots,b_{1})$, where $r=\log_2n$. Assume that $wt(i)$ denotes the number of ones in the binary expansion of $i$. Denote the sequence $(u_1,\ldots,u_{n})$ by ${\mathbf{u}}_1^{n}$ over $\mathbb{F}_q$ and its binary vector representation by $\mathbf{b}(\mathbf{u}_1^{n})=(\mathbf{b}({u}_1),\ldots,\mathbf{b}({u}_n))$. Suppose that the superscript $T$ stand for the transpose of a vector. %{\color{blue}{Let ``$\oplus$'' and ``$\otimes$'' represent the addition over $\mathbb{F}_2$ and the Kronecker product, respectively, and the superscript $T$ stand for the transpose of a vector. In this paper, we suppose that ``$+$" represent the addition over the field, which is consistent with the two addends, i.e., if the two addends belongs to $\mathbb{F}_q$, then ``$+$" represent the addition over $\mathbb{F}_q$.}}
Let ``$\oplus$'' represent the addition modulo 2 and ``$+$" represent the addition over the finite field or real field, depending on the two addends, i.e., if the two addends belong to $\mathbb{F}_q$, then ``$+$" represents the addition over $\mathbb{F}_q$.

\section{Two-Stage Polarization based Nonbinary Polar Codes}
In this section, we first introduce the generator matrix, which is constructed by multiplicative repetitions for \textit{symbol-level} polarization. Then, the \textit{bit-level} polarization is performed by using a linear transformation for each input symbol, where the invertible matrix is a binary polarization matrix. Based on the two-stage polarization, \textit{bit-level} code construction is provided via the genie-aided symbol-based SC decoding method. According to the code construction, the corresponding decoding algorithm is also given. The construction processing as well as the encoder structure of two-stage polarization-based nonbinary polar codes is shown in Fig. \ref{fig:1}.
\begin{figure*}[!t]
	\centering
	\includegraphics[width=5.2in]{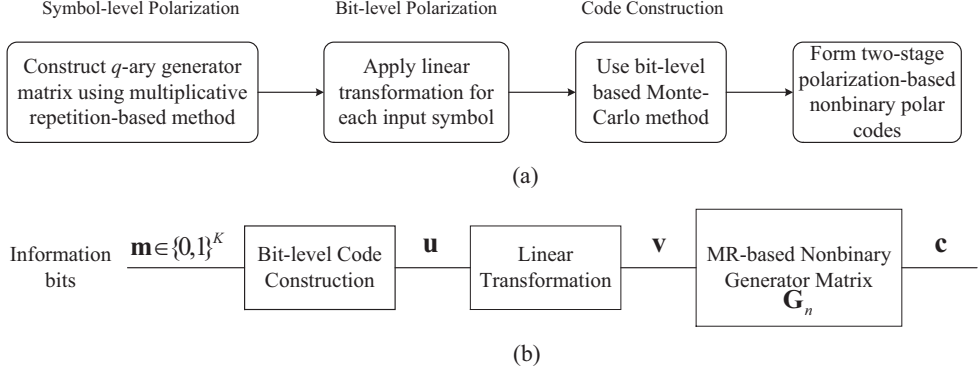}
	\caption{The construction processing and encoder structure of two-stage polarization-based nonbinary polar codes.}
	\label{fig:1}
\end{figure*}
\subsection{Multiplicative Repetition-based Matrix for Symbol-level}
We extend binary polar codes to nonbinary polar codes by considering a
$2\times2$ $q$-ary kernel given by \[\mathbf{F}=\left[{\begin{array}{*{20}{c}}\tiny
1&0\\
{\alpha^i}&1
\end{array}}\right],\]
where $\alpha^i\in \mathbb{F}^*_q$. With $\mathbf{F}$ as the generator matrix, for each ${\mathbf{u}}_1^2\in$ $\mathbb{F}^2_q$, we have the coded sequence  ${\mathbf{c}}_1^2={\mathbf{u}}_1^2\mathbf{F}=(u_1+\alpha^i u_2, u_2)$. In the case of $\alpha^i=1$, $\mathbf{F}$ has the same form as Ar{\i}kan's, and $u_2$ is repeated once and then superimposed on $u_1$.
In the nonbinary case with $\alpha^i\neq 1$, $\alpha^i u_2$ is a multiplicative repetition of $u_2$. For this reason, we call this \textit{multiplicative repetition} (MR)-based matrix for convenience.

Let $\mathbf{F}^{\otimes 0}\triangleq[1]$ and $\alpha ^{{i_r}}$ represent the multiplier for the $r$-th Kronecker operation. The generator matrix with the code length $n=2^r$ is given by
\begin{equation}
{\mathbf{G}_n}=\mathbf{F}^{\otimes r}=\left[ {\begin{array}{*{20}{c}}
	{\mathbf{F}^{\otimes(r-1)}}&0\\
	{{\alpha ^{{i_r}}}{\mathbf{F}^{\otimes(r-1)}}}&{{\mathbf{F}^{\otimes(r-1)}}}
	\end{array}} \right],~r\geq 1,
\end{equation}
where all multipliers $\alpha ^{{i_j}}$ ($1\leq j\leq r$) are nonzero elements in $\mathbb{F}_q^*$. Fig. \ref{fig:2} shows a Forney-style factor graph of $\mathbf{G}_8$.
%${\mathbf{G}_n}$ is $([1,\alpha^{i_r}]\otimes[1,\alpha^{i_{r-1}}]\otimes\cdots\otimes[1,\alpha^{i_1}])^T$. 
\begin{figure}[t]
	\centering
	\includegraphics[width=3.36in]{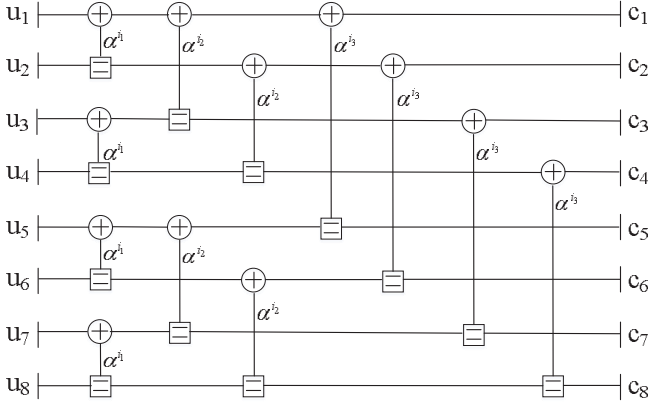}
	\caption{Forney-style factor graph of $\mathbf{G}_8$.}
	\label{fig:2}
\end{figure}

\begin{pro}\label{pro2}
Any $q$-ary input-channels can be polarized by the proposed matrix $\mathbf{G}_n$ over $\mathbb{F}_q$ with the multiplier
\begin{equation}
\alpha^{i_j}=\alpha^{2^{j-1}},~~1\leq j\leq r.
\end{equation}
\end{pro}

\begin{proof}
It can be seen that the multiplier $\alpha^{i_j}$ for any $1\leq j\leq r$ is a root of the primitive polynomial $f(x)$ over $\mathbb{F}_q$, then $\mathbb{F}_2(\alpha^{i_j})=\mathbb{F}_q$. By Proposition \ref{pro1}, we have the claim.
%We can see that ${\mathbf{G}_n}$ is a lower triangular matrix with unit diagonal elements, i.e., a standard form matrix. Due to all nonzero diagonal elements, the proposed lower triangular matrix $\mathbf{G}_n$ is an invertible matrix. It can also be seen that the multiplier $\alpha^{i_j}$ for any $1\leq j\leq r$ is a root of the primitive polynomial $f(x)$ over $\mathbb{F}_q$, then $\mathbb{F}_2(\alpha^{i_j})=\mathbb{F}_q$, resulting in $\mathbb{F}_2({\mathbf{G}_n})=\mathbb{F}_q$. Thus, by Theorem 1, we have the claim.
\end{proof}
%Note that, if $r\geq m$, then the elements in the first column of the ${\mathbf{G}_n}$ in (1) contains all elements of $\mathbb{F}^*_q$. 
With the advantage of similar structures to binary polar codes, the $i$-th ($1\leq i\leq n$) row-weight of the proposed nonbinary polar code also equals $wt(i-1)$. Inspired by the lemma in \cite{Zhang2017} for binary cases, we extend it to the proposed $q$-ary polar codes.

\begin{lem}
\label{lem1}
For a sequence ${\mathbf{u}}_1^{n}\in\{({\mathbf{u}}_1^{i-1}=\mathbf{0},{\mathbf{u}}_i\in \mathbb{F}_q^*,{\mathbf{u}}_{i+1}^{n}\in \mathbb{F}_q^{n-1-i})\}$ with $i=1,\ldots,n$, the weight of the coded sequence $d(\mathbf{c}_1^{n})$ with $\mathbf{c}_1^{n}={\mathbf{u}}_1^{n}\mathbf{G}_n$ is greater than or equal to the $i$-th row-weight in $\mathbf{G}_n$, i.e., $d(\mathbf{c}_1^{n})\geq 2^{wt(i-1)}$.
\end{lem}

\begin{proof}
For the case that $n=2$, the generator matrix is 
$\mathbf{G}_2=\mathbf{F}=\left[{\begin{array}{*{20}{c}}\tiny
1&0\\
{\alpha^{i_1}}&1
\end{array}}\right]$ and one can check that this lemma is true.

Assume that the lemma is true for $n=2^{r-1}$ with $r\geq 2$. For the case of $n=2^r$, we have \[\mathbf{G}_n=\mathbf{F}^{\otimes r}=\left[{\begin{array}{*{20}{c}}\tiny
\mathbf{G}_{n/2}&\mathbf{0}\\
{\alpha}^{i_r}\mathbf{G}_{n/2}&\mathbf{G}_{n/2}
\end{array}}\right]\] and $\mathbf{u}_1^{n}=(\mathbf{u}_1^{n/2},\mathbf{u}_{n/2+1}^{n})$.
Thus, $\mathbf{c}_1^{n}=(\mathbf{u}_1^{n/2}\mathbf{G}_{n/2}+\mathbf{u}_{n/2+1}^{n}{\alpha}^{i_r}\mathbf{G}_{n/2}, \mathbf{u}_{n/2+1}^{n}\mathbf{G}_{n/2})$.
\begin{itemize}
\item When $i\geq n/2$, we have $\mathbf{c}_1^{n}=(\mathbf{u}_{n/2+1}^{n}{\alpha}^{i_r}\mathbf{G}_{n/2},$ $\mathbf{u}_{n/2+1}^{n}\mathbf{G}_{n/2})$. Due to $\alpha^{i_r}\in \mathbb{F}^*$, thus $d(\mathbf{u}_{n/2+1}^{n}{\alpha}^{i_r}\mathbf{G}_{n/2})=d(\mathbf{u}_{n/2+1}^{n}\mathbf{G}_{n/2})$, resulting in $d(\mathbf{c}_1^{n})=2d(\mathbf{c}_1^{n/2})$. From the assumption that $d(\mathbf{c}_1^{n/2})\geq 2^{wt(i-1)}$ with $1\leq i\leq n/2$, we have $d(\mathbf{c}_1^{n})\geq 2\times2^{wt(i-n/2-1)}=2^{wt(i-1)}$ with $n/2+1\leq i\leq n$.
\item When $i< n/2$, we have $d(\mathbf{c}_1^{n})=d(\mathbf{u}_1^{n/2}\mathbf{G}_{n/2}+\mathbf{u}_{n/2+1}^{n}{\alpha}^{i_r}\mathbf{G}_{n/2})+d(\mathbf{u}_{n/2+1}^{n}\mathbf{G}_{n/2})$. Suppose that there are $s$ identical non-zero elements in $\mathbf{u}_1^{n/2}\mathbf{G}_{n/2}$ and $\mathbf{u}_{n/2+1}^{n}{\alpha}^{i_r}\mathbf{G}_{n/2}$ at the same position simultaneously. Then, $d(\mathbf{c}_1^{n})=d(\mathbf{u}_1^{n/2}\mathbf{G}_{n/2})+2d(\mathbf{u}_{n/2+1}^{n}\mathbf{G}_{n/2})-2s$. Clearly,
$s\leq d(\mathbf{u}_{n/2+1}^{n}\mathbf{G}_{n/2})$, thus we have $d(\mathbf{c}_1^{n})\geq d(\mathbf{u}_1^{n/2}\mathbf{G}_{n/2})\geq 2^{wt(i-1)}$ with $1\leq i\leq n/2$.
\end{itemize}
Thus, we complete the proof of Lemma 1.
\end{proof}

\begin{thm}
The minimum (symbol-wise) Hamming distance of the proposed nonbinary polar code $\mathcal{C}$ is given by $d_{\text{min}}(\mathcal{C})=\min_{i\in \mathcal{A}}2^{wt(i-1)}$.
\end{thm}

\begin{proof}
On the one hand, $d_{\text{min}}(\mathcal{C})$ cannot be larger than the minimum row-weight of the generator matrix, i.e., $d_{\text{min}}(\mathcal{C})\leq\min_{i\in \mathcal{A}}2^{wt(i-1)}$. On the other hand, any sequence ${\mathbf{u}}_1^{n}$ excluding all-zero vector belongs to $\{({\mathbf{u}}_1^{i-1}=\mathbf{0},{\mathbf{u}}_i\in \mathbb{F}_q^*,{\mathbf{u}}_{i+1}^{n}\in \mathbb{F}_q^{n-1-i})\}$ with $i=1,\ldots,n$. From Lemma \ref{lem1}, for $\mathbf{c}_1^{n}$ $\text{with}~u_{i\in\mathcal{A}}\in\mathbb{F}_q,{u}_{i\in\mathcal{A}^c}=\mathbf{0}$, we have $d(\mathbf{c}_1^{n})\geq \min_{i\in \mathcal{A}}\{2^{wt(i-1)}\}$.
Since the nonbinary polar codes with $\mathbf{G}_{n}$ are linear codes, thus $d_{min}(\mathcal{C})\geq \min_{i\in \mathcal{A}}\{2^{wt(i-1)}\}$.
Therefore, we have $d_{\text{min}}(\mathcal{C})=\min_{i\in \mathcal{A}}2^{wt(i-1)}$.	
\end{proof}	
%Based on the above discussion, it can be seen that, for a given $\mathcal{A}$, the minimum Hamming distance is not relevant to the choices of nonzero multipliers. While, the weight distribution heavily relies on the multipliers, which affects the error-correcting performance. In this paper, the multipliers are selected according to Proposition \ref{pro2}, and we refer to nonbinary polar codes constructed by MR-based matrix as MR-based nonbinary polar codes.
From the above discussion, we see that, for a given $\mathcal{A}$, the choices of nonzero multipliers do not affect the minimum Hamming distance. However, we need point out that the choices do have impact on the weight distribution, which in turn affect the error-correcting performance.
In this paper, the multipliers are selected according to Proposition \ref{pro2}, and we refer to nonbinary polar codes constructed by MR-based matrix as MR-based nonbinary polar codes.

\subsection{Linear Transformation Construction}
In the pioneer works~\cite{Park2012, Wu2016}, multiple polarization exists for $q=2^r$ polar codes showing multiple capacities from $1$ to $r$ bits in a symbol. To make full use of this phenomenon and intensify the bit-level polarization, a linear transformation (LT) is applied to each input symbol $\mathbf{u}_1^n$ before nonbinary polar encoding, resulting in symbols $\mathbf{v}_1^n$ with $v_i=\mathcal{T}_m(u_i)\in \mathbb{F}_q$ ($1\leq i\leq n$). To the best of our knowledge, this is the first attempt to construct nonbinary polar codes in terms of bit-level polarization.
From Definitions 1 and 2, we have $\mathbf{c}_1^n=(\mathcal{T}_m(u_1),\ldots,\mathcal{T}_m(u_n))\mathbf{G}_n$ and $\mathbf{b}(\mathbf{c}_1^n)=\mathbf{b}(\mathbf{u}_1^n)\mathcal{T}_m(\mathbf{\tilde{G}}_b)$, where $\mathcal{T}_m(\mathbf{\tilde{G}}_b)$ represents the $mn\times mn$ binary generator matrix obtained by replacing each $\alpha^i$ of $\mathbf{G}_n$ with the binary matrix $\mathbf{H}_m\mathbf{A}^i$.
The matrix $\mathbf{H}_m$ is an $m\times m$ binary polarization matrix chosen from~\cite{Lin2015}. The matrices $\mathbf{H}_m$ used in this paper for different $q$ are shown in Appendix A. With the given $\mathbf{H}_m$, the bit-level polarization phenomenon holds for the equivalent generator matrix $\mathcal{T}_m(\mathbf{\tilde{G}}_b)$, since it satisfies that none of its column permutation is an upper triangle matrix~\cite{Lin2015}. Specifically, we have the following propositions.

Let $W:\mathbb{F}_q\rightarrow\mathcal{Y}$ be a $q$-ary DMC and $W^n:\mathbb{F}_q^n\rightarrow\mathcal{Y}^n$ denote the vector channel defined by $W^n(\mathbf{y}_1^n|\mathbf{c}_1^n)=\prod_{i=1}^{n}W(y_i|c_i)$ with $c_i\in\mathbb{F}_q$.
Similar to the binary case, we say that a pair of $q$-ary input channels $W': \mathbb{F}_q\rightarrow\mathcal{Y}^2$ and $W'': \mathbb{F}_q\rightarrow\mathcal{Y}^2\times \mathbb{F}_q$ are obtained by combining two independent copies of $W$ with $c_1=v_1+\alpha v_2$ (the multiplier is selected from (1) to ensure polarization) and $c_2=v_2$ as the input to each copy, where
\begin{equation}
\begin{split}
&W'(y_1,y_2|v_1)=\frac{1}{q}\sum_{v_2\in\mathbb{F}_q}W(y_1|v_1+\alpha v_2)W(y_2|v_2)\\
&W''(y_1,y_2,v_1|v_2)=\frac{1}{q}W(y_1|v_1+\alpha v_2)W(y_2|v_2).
\end{split}
\end{equation}
By the chain rule, we have $I(W')+I(W'')=2I(W)$. Furthermore, we also have $I(W')\leq I(W'')$ with equality if $I(W)$ equals $0$ or $r$.

Consider bit-level processing and assume that each $c_i$ is transmitted over a set of $m$ independent binary input channels. Let $\tilde W: \mathbb{F}_2\rightarrow\mathcal{\tilde Y}$ be a binary DMC and $f:(\tilde y_{i,1},\ldots,\tilde y_{i,m})\in\mathcal{\tilde Y}^m \mapsto y_i\in\mathcal{Y}$ be a one-to-one mapping, then we have 
\[W(y_i|c_i)=W(f(\tilde{y}_{i,1},\ldots,\tilde{y}_{{i,m}})|\mathbf{b}(c_i))=\prod_{j=1}^{m}\tilde W(\tilde{y}_{i,j}|c_{i,j}),\] 
i.e., $W\rightarrow\underbrace{(\tilde W,\ldots,\tilde W)}_\text{\emph{m}}$. From $v_i=\mathcal{T}_m(u_i)$, we have $(c_{1,1},\ldots,c_{1,m})=\mathbf{b}(v_1+\alpha v_2)=\mathbf{b}(\mathcal{T}_m(u_1)+\alpha \mathcal{T}_m(u_2))$ and $(c_{2,1},\ldots,c_{2,m})=\mathbf{b}(v_2)=\mathbf{b}(\mathcal{T}_m(u_2))$. Consider $2m$ independent copies of $\tilde W$ with $c_{1,1},\ldots,c_{1,m},c_{2,1},\ldots,c_{2,m}$ as the input to each copy. Then
let us specify the channels as follows:
\begin{equation*}
\left\{{
\begin{aligned}
&{W'_{1}:\mathbb{F}_2\rightarrow\mathcal{Y}^2}\\
&{W'_{2}:\mathbb{F}_2\rightarrow\mathcal{Y}^2\times\mathbb{F}_2}\\
&{~~~~~~~~~~~\vdots}\\
&{W'_{m}:\mathbb{F}_2\rightarrow\mathcal{Y}^2\times\mathbb{F}_2^{m-1}}\\
&{W''_{1}:\mathbb{F}_2\rightarrow\mathcal{Y}^2\times\mathbb{F}_2^m}\\
&{W''_{2}:\mathbb{F}_2\rightarrow\mathcal{Y}^2\times\mathbb{F}_2^{m+1}}\\
&{~~~~~~~~~~~\vdots}\\
&{W''_{m}:\mathbb{F}_2\rightarrow\mathcal{Y}^2\times\mathbb{F}_2^{2m-1}}\\
\end{aligned}
}\right.,
\end{equation*}
i.e., $(W, W)\rightarrow \underbrace{(\tilde W,\ldots,\tilde W)}_\text{2\emph{m}}\rightarrow (W'_{1},\ldots, W'_{m}, W''_{1},\ldots,\\ W''_{m})\rightarrow (W', W'')$.

For example, consider $q=4$ and $\mathbf{H}_2=\left[ {\begin{smallmatrix}
1&0\\
{1}&{1}
\end{smallmatrix}}\right]$ over $\mathbb{F}_2$, we have $\mathbf{b}(v_i)=(v_{i,1},v_{i,2})=\mathbf{b}(u_i)\mathbf{H}_2=(u_{i,1}\oplus u_{i,2}, u_{i,2})$, i.e., $v_i=g(u_{i,1}\oplus u_{i,2}, u_{i,2})$ ($i=1,2$),
then the channels $W'_{1}$, $W'_{2}$, $W''_{1}$ and $W''_{2}$ is defined as
\begin{equation}
\begin{aligned}
&W'_{1}(y_1,y_2|u_{1,1})=W'_{1}(f(y_{1,1},y_{1,2}),f(y_{2,1},y_{2,2})|u_{1,1})\\
&~~~~~~~~~~~~~~~~~~~=\frac{1}{2^3}\sum_{\substack{u_{1,2}\in\mathbb{F}_2\\u_{2,1}\in\mathbb{F}_2\\u_{2,2}\in\mathbb{F}_2}}W(f(y_{1,1},y_{1,2})|g(u_{1,1}\oplus u_{1,2},\\
&~~~~~~~~~~~~~~~~~~~~~~~~u_{1,2})+\alpha g(u_{2,1}\oplus u_{2,2}, u_{2,2}))W(f(y_{2,1},\\
&~~~~~~~~~~~~~~~~~~~~~~~~y_{2,2})|g(u_{2,1}\oplus u_{2,2}, u_{2,2}))\\
&~~~~~~~~~~~~~~~~~~~=\frac{1}{2^3}\sum_{u_{1,2}\in\mathbb{F}_2}\sum_{v_{2}\in\mathbb{F}_4}W(y_1|v_1+ \alpha v_2)W(y_2|v_2)\\
&~~~~~~~~~~~~~~~~~~~=\frac{1}{2}\sum_{u_{1,2}\in\mathbb{F}_2}W'(y_1,y_2|v_1)\\
&W'_{2}(y_1,y_2,u_{1,1}|u_{1,2})=\frac{1}{2}W'(y_1,y_2|v_1)\\
&W''_{1}(y_1,y_2,v_1|u_{2,1})=W'_{1}(f(y_{1,1},y_{1,2}),f(y_{2,1},y_{2,2}),v_1|u_{2,1})\\
&~~~~~~~~~~~~~~~~~~~~~~~=\frac{1}{2}\sum_{u_{2,2}\in\mathbb{F}_2}W''(y_1,y_2,v_1|v_2)\\
&W''_{2}(y_1,y_2,v_1,u_{2,1}|u_{2,2})=\frac{1}{2}W''(y_1,y_2,v_1|v_2).
\end{aligned}
\end{equation}

\begin{pro}
Suppose $(\tilde W,\tilde W,\tilde W,\tilde W)\rightarrow(W'_{1},W'_{2},$
$W''_{1},W''_{2})$
for a set of binary-input channels and $(W,W)\rightarrow(W',W'')$ for a set of quaternary-input channels. Then
\begin{equation}
{I(W'_{1})+I(W'_{2})=I(W')}
\end{equation}
\begin{equation}
{I(W''_{1})+I(W''_{2})=I(W'')}
\end{equation}
\begin{equation}
{I(W'_{1})+I(W'_{2})+I(W''_{1})+I(W''_{2})=4I(\tilde W)=2I(W)}
\end{equation}
\begin{equation}
{I(W'_{1})\leq I(W'_{2})}
\end{equation}
\begin{equation}
{I(W''_{1})\leq I(W''_{2})}
\end{equation}
with equality if $I(\tilde W)=0$ or $1$.
\end{pro}
\begin{proof}
The proof is given in Appendix B.
\end{proof}

Proposition 1 can be extended to arbitrary $q$-ary polar codes with an $m\times m$ binary polarization matrix $\mathbf{H}_m$.
\begin{pro}
Consider a single-step transformation
of two independent copies of a $q$-ary-input. Define
\begin{equation}
\begin{aligned}
&W'_{j}(y_1,y_2,\mathbf{b}_1^{\tilde j}(u_1)|u_{1,j})=\frac{1}{2^{m-1}}\sum_{\substack{u_{1,t}\in\mathbb{F}_2\\j<t\leq m}}W'(y_1,y_2|v_1)\\
&W''_{j}(y_1,y_2,v_1,\mathbf{b}_1^{\tilde j}(u_2)|u_{2,j})=\frac{1}{2^{m-1}}\sum_{\substack{u_{2,t}\in\mathbb{F}_2\\j<t\leq m}}W''(y_1,y_2,v_1|v_2),
\end{aligned}
\end{equation}
where $\tilde j=j-1$, and $v_i=\mathcal{T}_m(u_i)\in \mathbb{F}_q, i=1,2$.
Suppose $\underbrace{(\tilde W,\ldots,\tilde W)}_\text{2\emph{m}}\rightarrow (W'_{1},\ldots,W'_{m},W''_{1},\ldots,W''_{m})$
for some set of binary-input channels and $(W,W)\rightarrow(W',W'')$ for some set of $q$-ary-input channels. Then
\begin{equation}
{\sum_{j=1}^{m}I(W'_{j})=I(W')}
\end{equation}
\begin{equation}
{\sum_{j=1}^{m}I(W''_{j})=I(W'')}.
\end{equation}
\end{pro}

With the definition of binary polarization matrix, there exists two different $W'_{j}$ ($W''_{j}$) having different capacities, which implies that bit-level construction can be considered when $I(\tilde W)$ is not equal to $0$ nor $1$.

Note that the concept of two-stage polarization can also be extent to other polarization matrix based nonbinary polar codes, such as RS-based polar codes and Hermitian-based polar codes. In this paper, we only consider two-stage polarization of MR-based matrix.

\subsection{Bit-level Computation-based Code Construction}
With two-stage polarization, we consider the equivalent channel reliabilities on bit-level, which is different from the conventional nonbinary code construction with sorting symbol-channel reliabilities.

Similar to the Monte-Carlo approach for binary cases, a genie-aided symbol-based SC decoder is used to compute the bit-channel reliabilities for $q$-ary polar codes with $\mathbf{G}_n$, where $N$ information bits are uniformly generated resulting in $n$ $q$-ary symbols $\mathbf{u}_1^n$. After LT construction, $n$ $q$-ary input-symbols $\mathbf{v}_1^n$ are delivered into $q$-ary polar encoder, where $v_i=\mathcal{T}_m(u_i)$ for each $0\leq i\leq n$. Suppose that the bit-reversal $\pi(\cdot)$ is used at the receiver resulting in the vector $\mathbf{y}_1^n$. Let $p(y|v)$, $v\in \mathbb{F}_q$, represent the channel conditional probabilities given by the demapper. Due to the similar structure to binary polar codes, the recursive formulas for the symbol-based SC decoding are given as
\begin{equation}
\begin{aligned}
&p_\lambda ^{(2i-1)}(\mathbf{y}_1^\Lambda ,\mathbf{v}_1^{2i - 2}|{v_{2i-1}})\\
&=\sum\limits_{{v_{2i}}}{\frac{1}{q}\{p_{\lambda-1}^{(i)}(\mathbf{y}_1^{\Lambda /2},\mathbf{v}_{1,odd}^{2i-2} + {\alpha ^{{i_{t + 1}}}}\mathbf{v}_{1,even}^{2i - 2}|{v_{2i - 1}} + {\alpha ^{{i_t}}}{v_{2i}})}\\
&~~~~~\cdot p_{\lambda-1}^{(i)}(\mathbf{y}_{\Lambda /2 + 1}^\Lambda ,\mathbf{v}_{1,even}^{2i - 2}|{v_{2i}})\},
\end{aligned}
\label{eq:8}
\end{equation}
\begin{equation}
\begin{aligned}
&p_\lambda ^{(2i)}(\mathbf{y}_1^\Lambda,\mathbf{v}_1^{2i-1}|{v_{2i}})\\
&=\frac{1}{q}p_{\lambda - 1}^{(i)}(\mathbf{y}_1^{\Lambda /2},\mathbf{v}_{1,odd}^{2i-2}+{\alpha^{{i_{t + 1}}}}\mathbf{v}_{1,even}^{2i-2}|{v_{2i-1}} + {\alpha ^{{i_t}}}{v_{2i}})\\
&~~~~\cdot p_{\lambda-1}^{(i)}(\mathbf{y}_{\Lambda /2 + 1}^\Lambda,\mathbf{v}_{1,even}^{2i-2}|{v_{2i}}),
\end{aligned}
\label{eq:9}
\end{equation}
where $1 \leq \lambda\leq r={\log_2}n$, $\Lambda=2^{\lambda}$, $1\leq i\leq \lfloor\frac{\Lambda+1}{2}\rfloor$, $t=r-\lambda+1$ and $p_0^{(1)}(y|v) = p(y|v)$.

Denote the estimated symbols by $\hat{u}_1,\ldots,\hat{u}_n$. Then, for conventional \textit{symbol-level} computation, the decision rule is as follows:
\begin{equation}
\begin{split}
&\hat{u}_i=\mathcal{T}_m^{-1}(\beta),~\text{if}~p_r^{(i)}(\mathbf{y}_1^n, \mathbf{\hat{v}}_1^{i-1}|\beta)\ge p_r^{(i)}(\mathbf{y}_1^n,\mathbf{\hat{v}}_1^{i-1}|\gamma)\\
&~~~~~~~~~~~~~~~~~~\text{for~all}~\gamma ~\text{with}~\beta\neq \gamma\in\mathbb{F}_q.
\end{split}
\end{equation}
While for \textit{bit-level} computation, we first define the likelihood ratio $L_{i,j}$ for the $j$-th bit in $i$-th symbol,
\begin{equation}
L_{i,j}=\frac{\sum\limits_{\mathbf{b}(\mathcal{T}_m^{-1}(\beta))_j=0}p_r^{(i)}(\mathbf{y}_1^n,\mathbf{\hat{v}}_1^{i-1},\mathbf{b}_1^{\tilde{j}}(\hat{u}_i)|\beta)}{\sum\limits_{\mathbf{b}(\mathcal{T}_m^{-1}(\beta))_j=1}p_r^{(i)}(\mathbf{y}_1^n,\mathbf{\hat{v}}_1^{i-1},\mathbf{b}_1^{\tilde{j}}(\hat{u}_i)|\beta)},~~1\leq j\leq m,
\end{equation}
where $\tilde j=j-1$, and $\mathbf{b}(\mathcal{T}_m^{-1}(\beta))_j$ represents the $j$-th bit in the binary representation of the symbol that before LT. Then, the decision rule is given as:
\begin{equation}
\hat{u}_{i,j}=\left\{{\begin{array}{*{20}{c}}
	{0,~~\text{if}~L_{i,j}\geq 1}\\
	{1,~~\text{otherwise}}.
	\end{array}}\right.
\end{equation}
Thus, the estimate of $\hat{u_i}$ is given by $\hat{u_i}=g(\hat{u}_{i,1},\ldots,\hat{u}_{i,m})$.

Note that, with the genie, the reliability of ${(i,j)}$-th bit-channel $z_b({i,j})$ ($\forall 1\leq j\leq m$) is computed under the assumption that both symbols $\mathbf{u}_{1}^{{i-1}}$ ($1\leq i \leq n$) and $\tilde j=j-1$ bits $\mathbf{b}_1^{\tilde j}(u_i)$ in $i$-th symbol are available at the symbol-based SC decoder, which is different from the reliability computation of $i$-th symbol-channel $z_s(i)$ with the assumption that only symbols $\mathbf{u}_{1}^{{i-1}}$ ($1\leq i \leq n$) are available at the decoder. Compared to \cite{Arikan2009}, we calculate the error-rate of symbol-channels and bit-channels rather than their Bhattacharyya parameters to reflect the channel reliabilities.

Now, we consider the \textit{bit-level} construction, i.e., constructing the set $\mathcal{B}$ based on the descending order of \textit{bit-channel} reliabilities. Specifically, the set $\mathcal{B}$ consists of the indices of the lowest $|\mathcal{B}|$ elements in $\{\tilde{z}_b({i,j}), 1\leq i\leq n, 1\leq j\leq m\}$, where $|\mathcal{B}|$ is determined by the $K$ and check bits (if any).
The bit-channels indexed by $\mathcal{B}$ will be used to transmit unfrozen bits.
Due to the construction, frozen bits are typically distributed as shown in Fig. \ref{fig:10}(b). The set $\mathcal{A}$ is determined by $\mathcal{B}$, in which different unfrozen symbols may contain different numbers of unfrozen bits, denoted by $m_t$ with $1\leq m_t\leq m$.
%Particularly, each frozen symbol contains $m_t=0$ unfrozen bits.
Note that, for the conventional \textit{symbol-level} construction, each $q$-ary unfrozen symbol contains $m$ unfrozen bits.

\begin{figure}[!t]
	\centering
	\includegraphics[width=3.36in]{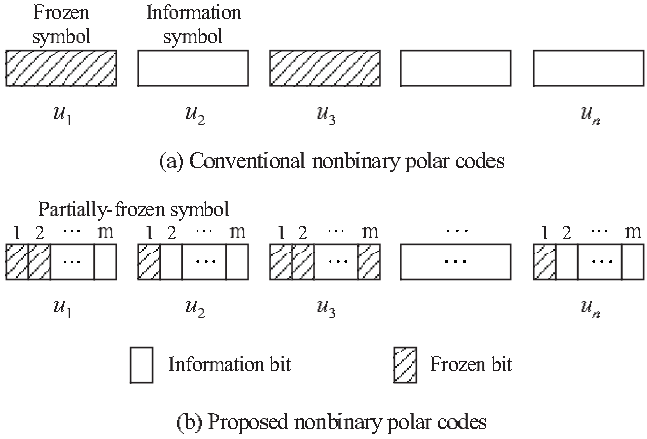}
	\caption{Input formats of the bits/symbols to the encoder.}
	\label{fig:10}
\end{figure}

In terms of $m_t$, we distinguish the symbols into three types, namely \textit{information symbols}, \textit{partially-frozen symbols} and \textit{frozen symbols}. Both information symbols and partially-frozen symbols are known as unfrozen symbols. With the \textit{bit-level} code construction, there is an advantage that the decoding search range\footnote{The decoding search range is referred to the set of decoding candidates for a symbol. For $q$-ary information symbols, the size of decoding search range is $q$, and for $q$-ary frozen symbols, the size is 1.} can be reduced at the decoder for the partially-frozen symbols, which will be discussed in the next section. In this paper, we use the element $0$ as the frozen symbol or bit.
%\begin{figure*}
%	\centering
%	\includegraphics[width=6.6in]{Fig4n.eps}
%	\caption{Comparison between conventional list decoding tree and irregular list decoding tree.}
%	\label{fig:4}
%\end{figure*}
\subsection{Decoding Algorithm}
According to the code construction, for different symbol types, the decoding search ranges are different. Let $\mathcal{B}_i$ represent the decoding search range of $i$-th ($1\leq i\leq n$) symbol, and the size $|\mathcal{B}_i|$ is determined by the unfrozen bits in $i$-th symbol. The computation of set $\mathcal{B}_i$ is given in Algorithm \ref{alg1}.
In the decoding tree, $i$-th node is split into $|\mathcal{B}_i|$ branches. Thus, the full decoding tree size is $\prod_{i=1}^n|\mathcal{B}_i|$ for the proposed $q$-ary polar codes, where $|\mathcal{B}_i|=1$ for $i\in\mathcal{A}^c$. Since different branches may exist among unfrozen nodes at decoding stages, we call this irregular decoding tree.

Different from the decision rule of $i$-th symbol $\hat{u}_i$ in Section III-C, the rule for SC-decoder here is
\begin{itemize}
\item for $i\in\mathcal{A}^c$, then $\hat{u}_i=0$.

\item for $i\in\mathcal{A}$, then $\hat{u}_i=\mathcal{T}_m^{-1}(\beta)\in\mathcal{B}_i,\text{if}~p_r^{(i)}(\mathbf{y}_1^n,\mathbf{\hat{v}}_1^{i-1}|\beta)\geq p_r^{(i)}(\mathbf{y}_1^n,\mathbf{\hat{v}}_1^{i-1}|\gamma)$ for all $\gamma$ with $\beta\neq \gamma$ and $\mathcal{T}_m^{-1}(\gamma)\in\mathcal{B}_i$.
\end{itemize}

Now we consider list decoder for the proposed $q$-ary polar codes. Assume that the maximum $L$ decoding paths are kept at each decoding stage. Denote the estimate $\hat{u}_i$ at path $l$ by $\hat{u}_i[l]$ ($1\leq l\leq L$). Note that $|\mathcal{B}_i|$ candidates should be considered for $\hat{u}_i[l]$ with any path $l$. Thus, there are a total of $L|\mathcal{B}_i|$ candidates corresponding to the $\hat{u}_i$. To maintain the maximum $L$ paths, if $L|\mathcal{B}_i|>L$, all $L|\mathcal{B}_i|$ candidate paths will be sorted according to their corresponding probabilities $p_r^{(i)}(\mathbf{y}_1^n,\mathbf{\hat{v}}_1^{i-1}|\beta)[l]$ with $\mathcal{T}_m^{-1}(\beta)\in\mathcal{B}_i$ and the paths with lowest probabilities will be pruned until only $L$ paths remain.
At the last stage, the decoder outputs the estimated information sequence
given by the decoding path with the largest probability.
The main SCL decoding algorithm for two-stage polarization-based polar codes is shown in Algorithm \ref{alg2}.

%As an example, a comparison between conventional list decoding tree (based on \textit{symbol-level} construction) and irregular list decoding tree (based on \textit{bit-level} construction) is shown in Fig. \ref{fig:4}, where $q=4$, $n=4$, $R=1/2$ and $L=4$. Assume that $\mathcal{A}=\{3,4\}$ with \textit{symbol-level} code construction, and $\mathcal{B}=\{4,5,6,8\}$ resulting in $\mathcal{A}=\{2,3,4\}$ with \textit{bit-level} code construction. The sets $\mathcal{B}_i$ with $1\leq i\leq 4$ are computed by Algorithm \ref{alg1}. For conventional list decoding, $qL=4L$ symbols are always considered for each unfrozen symbol, while for irregular decoding tree $|\mathcal{B}_2|L=|\mathcal{B}_4|L=2L$ symbols are considered for partially-frozen symbols. Discontinued paths are coloured white, and the final reserved paths are corresponded to the bit sequences $(0,0,0,0,0,0,0,0)$, $(0,0,0,0,0,0,1,0)$, $(0,0,0,0,0,1,1,0)$ and $(0,0,0,0,1,0,1,0)$.
\indent
\begin{algorithm}[!t]
	\caption{Computation~of~Set~$\mathbf{\mathcal{B}_i}$}
	\label{alg1}
	\KwIn{$\mathcal{A}^c$ and $\mathcal{B}^c$}
	\KwOut{Sets~$\{\mathbf{\mathcal{B}_i},1\leq i\leq n\}$}
	$tp\leftarrow 0$, ${P_1}\leftarrow \emptyset$, $\ldots~$, ${P_m}\leftarrow \emptyset$.\\
	\For {$i=1$ \text{to} $n$}
	{
		\If {$i \in \mathcal{A}^c$}
			{
				$\tiny{\bullet}$ Set $\mathcal{B}_i\leftarrow\{0\}$.
			}
		\Else{
			\For{$j=1$~to~$m$}{
			$tp\leftarrow m(i-1)+j$.\\
			\If{$tp\in\mathcal{B}^c$}
			{
			  ${P_j}\leftarrow \{0\}$.
		    }
			\Else{${P_j}\leftarrow \{0,1\}$.}
			}
		      $\tiny{\bullet}$ Set $\mathcal{B}_i\leftarrow P_1\times P_2\times \cdots\times P_m$ with $q$-ary element representation.
		}	
	}				
\end{algorithm}
\begin{algorithm}[!t]
\caption{SCL~Decoder~for~Nonbinary~Polar~Codes}
\label{alg2}
\KwIn{Sets~$\{\mathbf{\mathcal{B}_i},1\leq i\leq n\}$, $\mathcal{A}$, $\mathcal{A}^c$, and received channel probabilities}
\KwOut{Estimated infromation sequence}
\For {$i=1$ \text{to} $n$}
{
\For {$l=1$ \text{to} $L$}
{
\If {$i \in \mathcal{A}^c$}
{
 $\tiny{\bullet}$ Set $\hat{u}_i[l]=0$ and keep all paths.
}
\Else{
$\tiny{\bullet}$ Calculate the conditional probabilities 	$p_r^{(i)}(\mathbf{y}_1^n,\mathbf{\hat{v}}_1^{i-1}|\beta)[l]$ with $\mathcal{T}_m^{-1}(\beta)\in\mathcal{B}_i$ using Eq. (\ref{eq:8}) and Eq. (\ref{eq:9}).\\
}
}
\If{$i \in \mathcal{A}$}{
$\tiny{\bullet}$ Sort the conditional probabilities in descending order, and select the $L$ most likely paths with the largest probabilities.\\
}
}			
Find the most likely path with the largest probability.
\end{algorithm}
\subsection{Numerical Results}
Two examples of two-stage polarization-based polar codes are provided in this section. In all simulations, BPSK signaling over the AWGN channel is assumed.

\textbf{Example 1:}
Consider the MR-based nonbinary polar codes constructed by \textit{bit-level} computation, where two code rates $R=1/2$ and $R=1/3$ are considered for $q=16$, $N=2048$, and $L=8$. Refer to Fig. \ref{fig:9}, the performance of comparable binary polar codes and MR-based nonbinary polar codes constructed by \textit{symbol-level} computation are also given. For comparison, the $16$-ary polar codes with the generator matrix constructed by a pure multiplier are considered\footnote{Note that the elements $\alpha^2$, $\alpha^4$, and $\alpha^8$ are also the primitive elements of $\mathbb{F}_{16}$.}, which can be regarded as the 16-ary code using $\mathbf{G}_{\text{RS}}(16,2)$ mentioned in \cite{Mori20101}, and the codes are constructed by \textit{symbol-level} computation. All codes are constructed via the Monte-Carlo method at $E_b/N_0=2.0$ dB.

\begin{figure}[!t]
	\centering
	\includegraphics[width=3.36in]{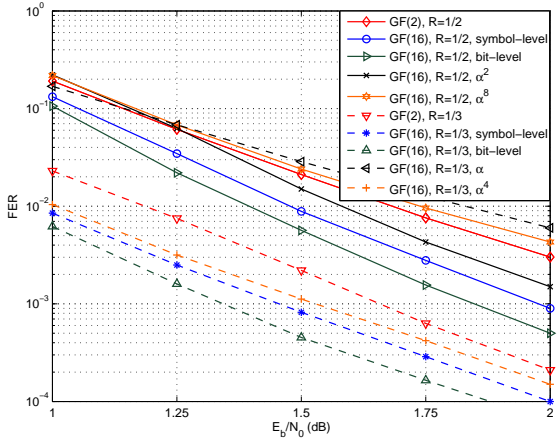}
	\caption{Performance comparison among binary polar codes, MR-based nonbinary polar codes with \textit{bit-level} computation and \textit{symbol-level} computation $N=2048$.}
	\label{fig:9}%($K=1024$, $K=684$)
\end{figure}

From Fig. \ref{fig:9}, it can be seen that the nonbinary polar code with a pure multiplier exhibits an inferior performance. We can also see that MR-based $16$-ary polar codes constructed with \textit{bit-level} computation perform better than that with \textit{symbol-level} computation for both $R=1/2$ and $R=1/3$. Moreover, the proposed $16$-ary polar codes can provide up to $0.35$ dB gain for $R=1/2$ at FER=$3\times10^{-3}$, and about $0.25$ dB gain for $R=1/3$ at FER=$2\times10^{-4}$, with respect to the binary polar codes.

\textbf{Example 2:}
\begin{figure}[!t]
	\centering
	\includegraphics[width=3.36in]{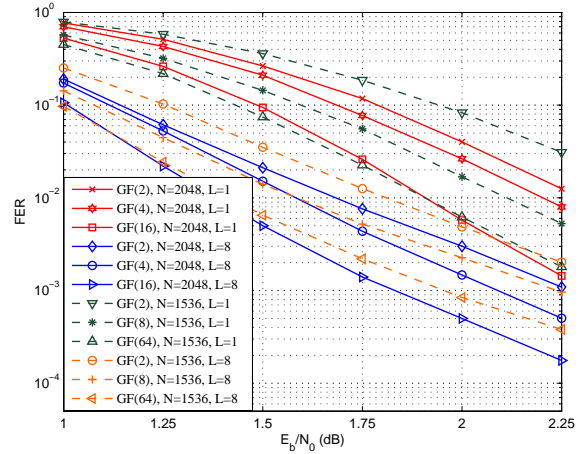}
	\caption{Performance comparison between binary polar codes and two-stage polarization-based nonbinary polar codes.}
	\label{fig:11}
\end{figure}
Consider MR-based nonbinary polar codes with \textit{bit-level} construction. The error-correcting performance of two equivalent code lengths $N=2048$ and $N=1536$ with $R=1/2$ are given, where both $q=4$ and $q=16$ are considered for $N=2048$, and both $q=8$ and $q=64$ are considered for $N=1536$. As a reference, the performance of comparable binary polar codes are also shown in Fig. \ref{fig:11}, in which the quasi-uniform puncturing (QUP) method in~\cite{Niu2013} is used to adapt the code length. All codes are constructed by the Monte-Carlo method at $E_b/N_0=2.0$ dB.

It can be seen that the proposed nonbinary polar codes exhibit better error-correcting performance than binary polar codes for both $L=1$ and $L=8$. Note that, for $N=1536$, non extra length-matching method is applied to nonbinary polar, which implies that larger code length range can be obtained by adapting the field order for nonbinary polar codes. Furthermore, with the increase of the field order, the performance can be improved.

\section{Improved Nonbinary Coding Methods}
In order to facilitate the SCL decoder to detect and prune error paths in time, the frozen bits in partially-frozen symbols are used as active-check bits, which is inspired by the concept of parity-check (or dynamic-frozen bits)~\cite{Trifonov2016}-\cite{Zhang2018}.
Moreover, similar to CRC-aided polar codes, a CRC outer code is also considered for the two-stage polarization-based nonbinary polar codes to improve the error-correcting performance.

\subsection{Active-check-used Nonbinary Polar Codes}
\subsubsection{Encoding Design}
According to \textit{bit-level} construction, it can be seen that the equivalent symbol-channels, which corresponds to partially-frozen symbols, exhibit inferior symbol-error probability leading to error-propagation, thus we pay more attention to them and attempt to use check constraint to reduce the error-propagation.

Assume that the set of the index of partially-frozen symbols is denoted by $\mathcal{A}^*$, where $\mathcal{A}^*\subseteq \mathcal{A}$. The construction is given by the following steps,
\begin{enumerate}
\item Find the frozen bit-channel with lowest bit-error probability among the frozen bit-channels in each partially-frozen symbol $t\in \mathcal{A}^*$, collectively denoted by $\mathcal{D}$ ($|\mathcal{D}|=|\mathcal{A}^*|$).

\item For each element in $\mathcal{D}$, construct the set $\mathcal{I}_i$ ($1\leq i\leq |\mathcal{D}|$), which is formed by the index $j\in\mathcal{B}$ with $j\leq m\mathcal{A}^*[i]$.

\item Generate $|\mathcal{D}|$ binary sequences $\mathbf{s}_i$ with length $|\mathcal{I}_i|$ randomly.
The sequence acts as a puncture pattern, and if the $j$-th element in $\mathbf{s}_i$ is $``0"$, then the $j$-th element in $\mathcal{I}_i$ will be deleted.

\item For each $i$, we obtain the active-check (ACK) bit $u_{{\hat i},j}=\biguplus\limits_{(mi'+j')\in\mathcal{I}_i} u_{i', j'}$, where $\hat i=\mathcal{A}^*[i]$, $j=\mathcal{D}[i]-m(\hat i-1)$, and $\biguplus$ denotes the modulo-2 sum.
\end{enumerate}

\textbf{Example 3:}
Let $q=16$, $n=8$, and $R=1/2$. Assume that $\mathcal{B}=\{14, 15, 16, 20, \ldots, 32\}$ (designed at $E_b/N_0=2.0$ dB with $|\mathcal{B}|=16$) resulting in $\mathcal{A}=\{4, 5, 6, 7, 8\}$ and $\mathcal{A}^*=\{4,5\}$. Under the assumption that $\mathcal{D}=\{13,19\}$ and according to the Step 2), we get $\mathcal{I}_1=\{14,15,16\}$ and $\mathcal{I}_2=\{14,15,16,20\}$. If $\mathbf{s}_1=\{0,1,1\}$ and $\mathbf{s}_2=\{1,0,1,1\}$, then the updated $\mathcal{I}_1=\{15,16\}$ and $\mathcal{I}_2=\{14,16,20\}$. Thus we have the active-check bits $u_{4,1}=u_{4,3}\oplus u_{4,4}$ and $u_{5,3}=u_{4,2}\oplus u_{4,4}\oplus u_{5,4}$, seen in Fig. \ref{fig:21}.

\begin{figure}[!t]
	\centering
	\includegraphics[width=3.36in]{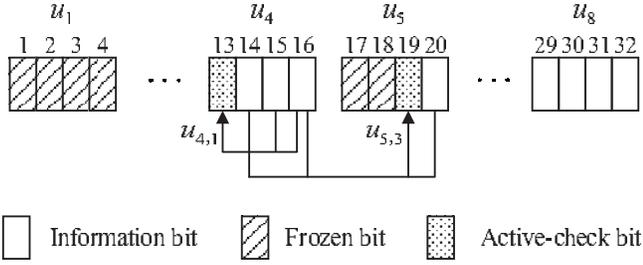}%11n
	\caption{The construction of active-check bits.}%$340$
	\label{fig:21}
\end{figure}

The set $\mathcal{I}_i$ exhibits the check constraint, which is based on a random generator in this paper, and can be optimized by other check approaches. Due to active-check bits, $\mathcal{B}^c=\{1,2,\ldots,N\}\backslash \{\mathcal{B\bigcup D\}}$.
\subsubsection{Decoding}
At the SCL decoder, each active-check bit is determined by the
estimated information bits which involve in the corresponding $\mathcal{B}_i$. Note that, with active-check bits, for different path $l$ the set $\mathcal{B}_i[l]$ of the partially-frozen symbol is different. Furthermore, although the decoding set $\mathcal{B}_i[l]$ changes, the size of $\mathcal{B}_i[l]$ does not change. The main SCL decoding with ACK is shown in Algorithm \ref{alg3}, where the decoding set $\mathcal{B}_i[l]$ is computed on-line and determined by the decoding.

\begin{algorithm}[!t]\footnotesize
\caption{SCL~Decoder~with~ACK}
\label{alg3}
\KwIn{$\mathcal{A}$, $\mathcal{A}^*$, $\mathcal{A}^c$, $\mathcal{B}$, $\mathcal{B}^c$ and received channel probabilities}
\KwOut{Estimated infromation sequence}
$tp\leftarrow 0$, $t\leftarrow 1$, $r\leftarrow 0$,${P_1}\leftarrow \emptyset$, $\ldots~$, ${P_m}\leftarrow \emptyset$, $\{\mathbf{\mathcal{B}_i}[L]\leftarrow \emptyset,1\leq i\leq n\}$.\\
\For {$i=1$ \text{to} $n$}
{
	\For {$l=1$ \text{to} $L$}
	{
		\If {$i \in \mathcal{A}^c$}
		{
			$\tiny{\bullet}$ Set $\mathcal{B}_i[l]\leftarrow \{0\}$.\\
			$\tiny{\bullet}$ Set $\hat{u}_i[l]=0$ and keep all paths.
		}
		\Else{
			\If{$i \in \mathcal{A}^*$}
			{
			
			\For{$j=1$~to~$m$}{
				$tp\leftarrow m(i-1)+j$.\\
				\If{$tp \in \mathcal{B}^c$}{
				  ${P_j}\leftarrow \{0\}$.
				}
				\Else{
				\If{$tp\in\mathcal{B}$}
				{
					${P_j}\leftarrow \{0,1\}$.
				}
				\Else{${P_j}\leftarrow$\text{Compute the active-check bit}\\ \text{according to} $\mathcal{B}_{t}$ \text{with the estimated} \text{information bits} $\mathbf{\hat{u}}_1^{i-1}[l]$ and $\{P_r, r\in\mathcal{B},\text{with}~m(i-1)+1\leq r\leq m\cdot i\}$.\\
				$t\leftarrow t+1$.\\
			}
			}
			}
			$\tiny{\bullet}$ Set $\mathcal{B}_i[l]\leftarrow P_1\times P_2\times \cdots\times P_m$ with $q$-ary element representation.
			
		    }
	
	        \Else{
	        $\tiny{\bullet}$ Set $\mathcal{B}_i[l]\leftarrow\{0,1,\dots,q-1\}$.
           }			
			$\tiny{\bullet}$ Calculate the conditional probabilities 	$p_r^{(i)}(\mathbf{y}_1^n,\mathbf{\hat{v}}_1^{i-1}|\beta)[l]$ with $\mathcal{T}_m^{-1}(\beta)\in\mathcal{B}_i[l]$ using Eq. (\ref{eq:8}) and Eq. (\ref{eq:9}).\\
		}
	}
	\If{$i \in \mathcal{A}$}{
		$\tiny{\bullet}$ Sort the conditional probabilities in descending order, and select the $L$ most likely paths with the largest probabilities.\\
	}
}			
Find the most likely path with the largest probability.
\end{algorithm}

\subsubsection{Numerical Results}
\begin{figure}[!t]
	\centering
	\includegraphics[width=3.36in]{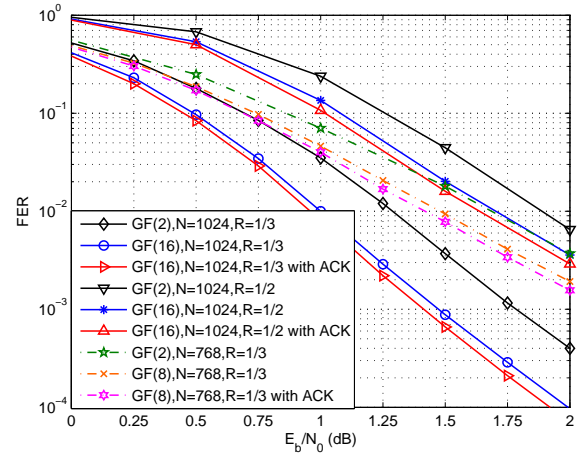}%11n
	\caption{Performance comparison among binary polar codes, nonbinary polar codes with/without ACK bits.}%$340$
	\label{fig:12}
\end{figure}
Performance comparison with BPSK over the AWGN channel among binary polar codes, nonbinary polar codes with/without ACK bits is shown in Fig. \ref{fig:12}, where two equivalent code lengths $N=1024$ and $N=768$ are considered under SCL method with $L=8$. All codes are constructed by Monte-Carlo method at $E_b/N_0=2.0$ dB. We can see that the proposed nonbinary polar codes with ACK outperform both the codes without ACK and binary polar codes for different $N$ and $R$.

\subsection{CRC-aided Nonbinary Polar Codes}
\subsubsection{Encoding and Decoding}
Similar to binary polar codes, a binary CRC outer code can also be concatenated with the proposed nonbinary polar codes  to further improve the error-correcting performance. The encoding and decoding structure with $t$-bit CRC is shown in Fig. \ref{fig:13}, where $K+t$ unfrozen bits are considered via \textit{bit-level} code construction. After SCL decoding, the decoder outputs the estimated information sequence given by the decoding path with the largest probability among the paths which can pass the CRC.
\begin{figure*}
	\centering
	\includegraphics[width=5.0in]{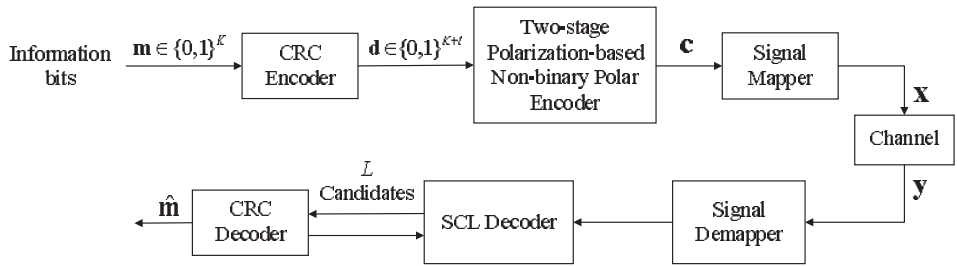} %fig19
	\caption{System model for CRC-aided nonbinary polar codes.}
	\label{fig:13}
\end{figure*}

\subsubsection{Numerical Results}
Two examples of the proposed polar codes with CRC-aided decoding are provided in this subsection, where CRC-8 is applied to all polar codes. In all simulations, BPSK signaling over the AWGN channel is assumed.

\textbf{Example 4:} Four rate-$1/2$ codes with equivalent code length of $N=512$ and $N=2048$ are simulated. For reference, the performance of comparable binary polar codes and RS4-based nonbinary polar codes, provided in \cite{Cheng2016}, are also given in~Fig. \ref{fig:18}. All codes are constructed by Monte-Carlo method at $E_b/N0=2.0$ dB.

\begin{figure}[!t]
	\centering
	\includegraphics[width=3.36in]{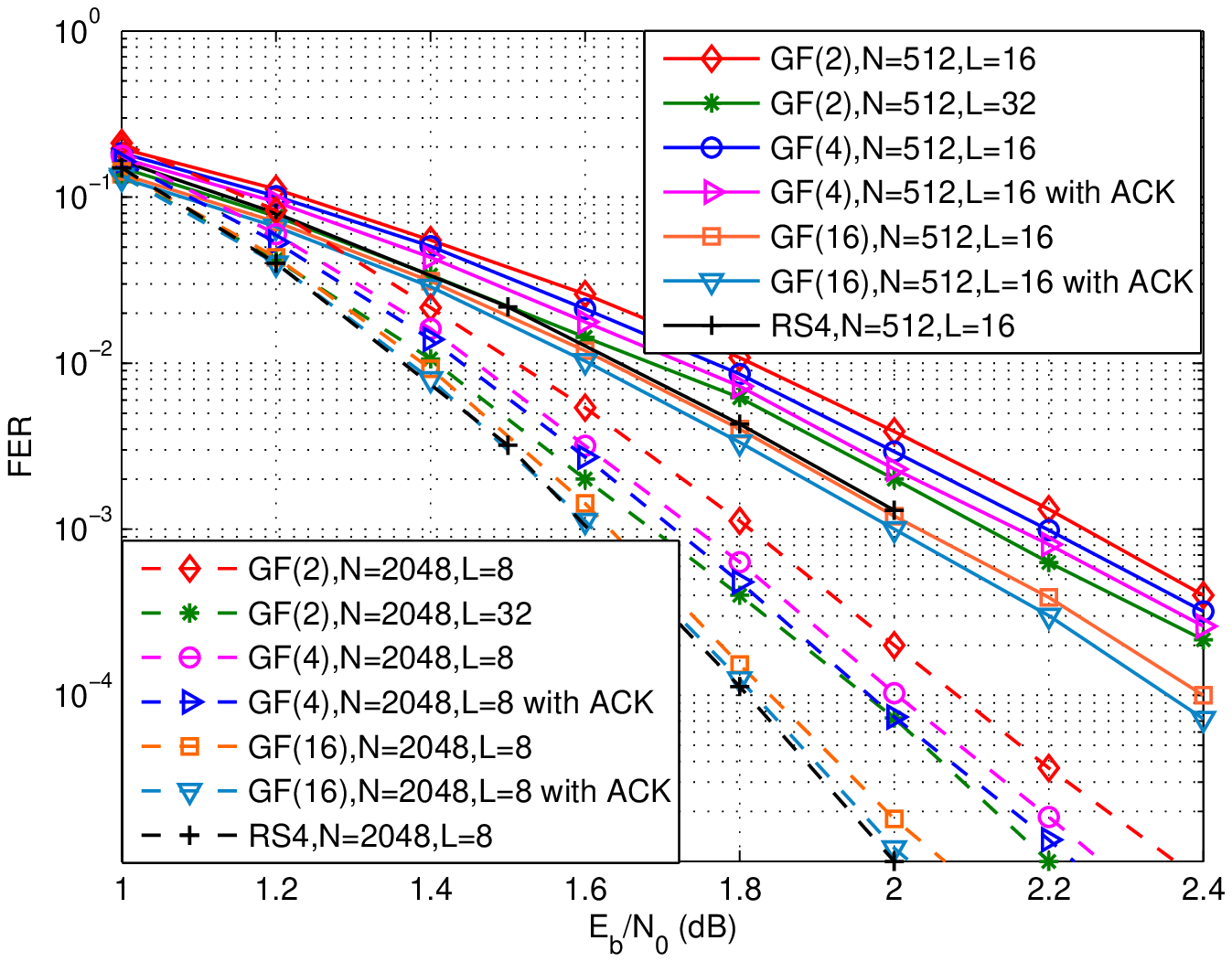}
	\caption{Performance comparison among binary polar codes, nonbinary polar codes with/without ACK bits under CRC-8-aided SCL method.}%$N=2048$($K=1024$, $K=684$)
	\label{fig:18}
\end{figure}

From Fig. \ref{fig:18}, it can be seen that the proposed nonbinary polar codes with ACK under CRC-aided SCL method also perform better than both the codes without ACK and binary polar codes. Moreover, compared with binary polar codes, similar performance can be obtained by the proposed polar codes with a smaller list size. Although the proposed $4$-ary polar codes exhibit inferior performance (about 0.1 dB) than $4$-ary RS-based nonbinary polar codes, the proposed codes have simple decoding structures and the performance can be improved by larger field orders.

\textbf{Example 5:}
Consider the performance comparison in URLLC. According to the simulation assumptions in \cite{Sybis2016}, three low code rates, $1/3$, $1/6$, and $1/12$ are considered, and the information lengths are all set to $K=256$. The $8$-ary polar codes with ACK are simulated for all rates, and all $8$-ary codes are constructed at $E_b/N_0=0.0$ dB. Binary polar codes are constructed by the Gaussian approximation (GA) method at $-1.59$ dB~\cite{Vangala2015}. The  CRC-aided list decoding with $L=8$ are applied to all polar codes. For reference, the performance of LDPC codes designed for URLLC in \cite{Wu2018} is also given, where the sum-product algorithm (SPA) with $20$ iterations is applied.
\begin{figure}[h]
	\centering
	\includegraphics[width=3.36in]{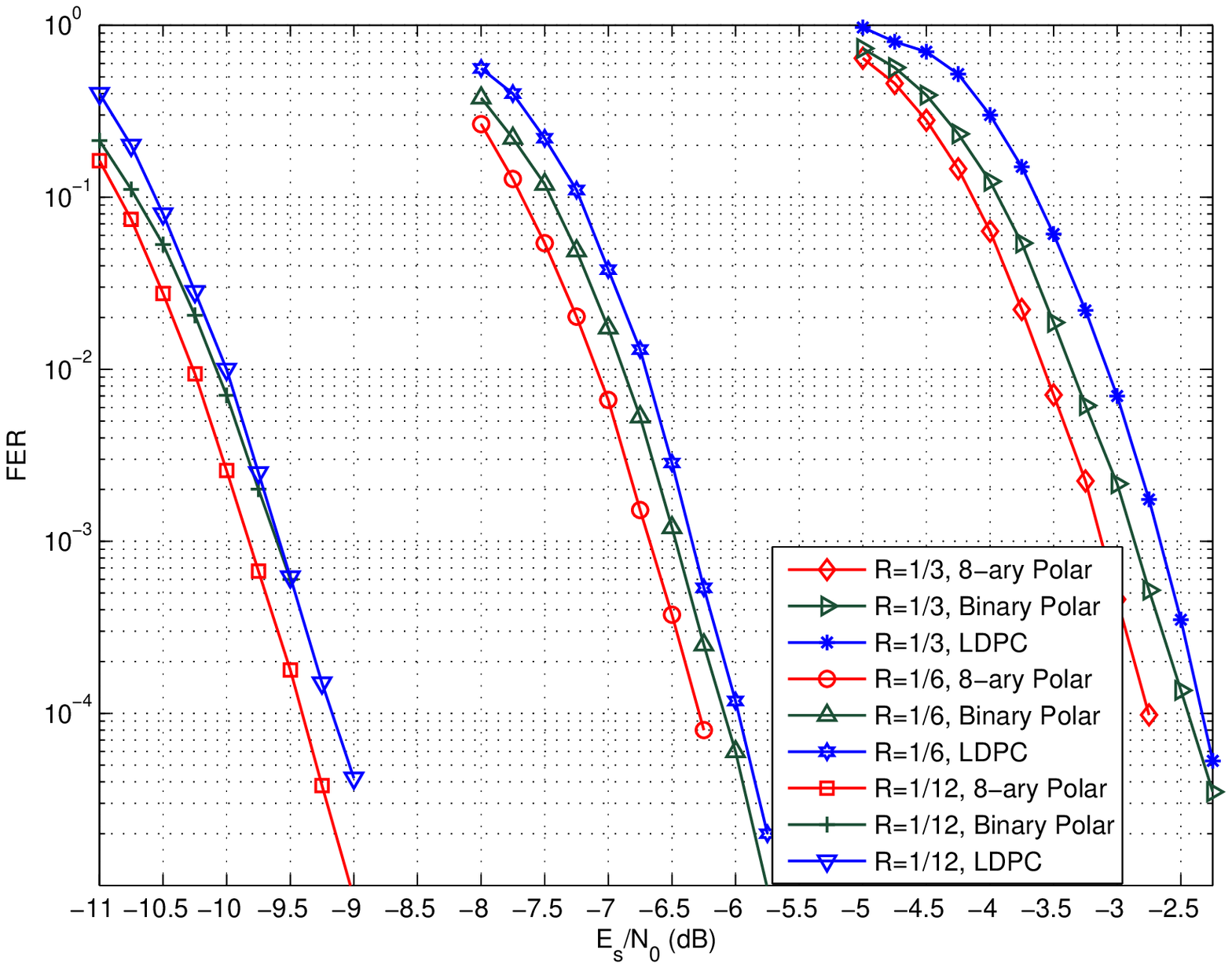}
	\caption{Performance comparison among binary polar codes, nonbinary polar codes with/without ACK bits under CRC-8-aided SCL method.}%$N=2048$($K=1024$, $K=684$)
	\label{fig:17}
\end{figure}

It can be seen that nonbinary polar codes outperform both binary polar codes  and LDPC codes in URLLC. In addition, with the increase of the code rate, the more coding gain can be obtained by the nonbinary polar codes.

\section{Analysis of Decoding}
\subsection{Decoding Latency}
Let us first consider the latency of SC decoding. Denote by functions $\mathbf{f}$ and $\mathbf{g}$ the node update equations (\ref{eq:8}) and (\ref{eq:9}), respectively. According to~\cite{Vardy20112}, for binary polar codes with\footnote{The rate-matching is considered, when $m$ is not a power of 2.} $N=2^{\lceil\log_2{m}\rceil} n$, $2N-2$ clock cycles (CLKs) are required for the SC decoder, where both functions $\mathbf{f}$ and $\mathbf{g}$ can be carried out in a single CLK. For the proposed $q$-ary polar codes, $(n-1)\log_2q$ CLKs are needed for the function $\mathbf{f}$ since $\log_2q$ additions are required for each unit $\mathbf{f}$, and the number of CLKs required for the function $\mathbf{g}$ is the same as that of binary cases. Thus, a total of $(n-1)\log_2q+(n-1)=(m+1)(n-1)$ CLKs are required for the proposed $q$-ary SC decoder.

A $L$-size SCL decoder can be viewed as the combination of $L$ copies of SC component decoders. In addition, the SCL decoder needs to sort $qL$ (for worst case) path metrics and selects the $L$ largest metrics for each decoded bit, thus, extra CLKs are required to carry out sorting and selecting functions. Using Bitonic sorter~\cite{Alexios2015}, the total number of stages is
$S=\frac{1}{2}\log_2(qL)(\log_2(qL)+1)$,
where each stage contains $\frac{qL}{2}$ compare-and-select (CAS) units consisting of one comparator and a 2-to-2 MUX. Generally, an intermediate variable (or register) is required to swap two numbers resulting in 3 CLKs~\cite{Han2017}. Since only unfrozen symbols require sorting and selecting functions, at most $3|\mathcal{A}|S=\frac{3}{2}|\mathcal{A}|(\log_2L+m)(\log_2L+m+1)$
CLKs are needed for $q$-ary polar decoder to sort and select paths. For binary decoder, the number of unfrozen bits is $|\mathcal{B}|$. The CLKs required by SC-decoder and SCL-decoder for binary and $q$-ary polar codes are tabulated in Table \ref{tab:1} for comparison.
\begin{table}[!t]\scriptsize
	\centering
	\renewcommand{\multirowsetup}{\centering}
	\renewcommand\arraystretch{1.96}
	\caption{Clock Cycles Comparison}
	\begin{tabular}{|p{1.16cm}<{\centering}|p{2.36cm}<{\centering}|p{3.96cm}<{\centering}|}%{|c|c|c|}
		\hline
		Polar codes&SC-decoder&SCL-decoder\\
		\hline
		$q$-ary&$(m+1)(n-1)$&$(m+1)(n-1)+\frac{3}{2}|\mathcal{A}|(\log_2L+m)(\log_2L+m+1)$\\
		\hline
		Binary&$2^{\lceil\log_2{m}\rceil+1}n-2$&$2^{\lceil\log_2{m}\rceil+1}n-2+\frac{3}{2}|\mathcal{B}|(\log_2L+1)(\log_2L+2)$ \\
		\hline
	\end{tabular}
	\label{tab:1}
\end{table}

\textbf{Example 6}: In Fig. \ref{fig:18}, with 8-bit CRC, for $N=512$, $|\mathcal{B}|$ is $264$, and $|\mathcal{A}|$ is $135$ for $\mathbb{F}_4$, and $68$ for $\mathbb{F}_{16}$, thus $8942$, $6096$ and $5742$ CLKs are needed for binary, 4-ary, 16-ary polar codes with SCL decoder, respectively, and for $N=2048$, $|\mathcal{B}|$ is $1032$, and $|\mathcal{A}|$ is $523$ for $\mathbb{F}_4$, and $264$ for $\mathbb{F}_{16}$, thus $35054$, $23562$ and $22216$ CLKs are required for binary, 4-ary, 16-ary polar codes with SCL decoder, respectively. It can be seen that the latency of the proposed nonbinary polar codes is lower than that of binary cases, and with the increase of the code length $n$, more CLKs can be saved by nonbinary cases. Furthermore, the latency can be decreased with the increase of the field order.
\subsection{Decoding Complexity}
The main decoding complexity of MR based nonbinary polar codes is shown in Table \ref{tab:2}. For better comparison, the complexity of list decoding, proposed by Vardy for binary polar codes, is also provided in Table \ref{tab:2}, where $N=2^{\lceil\log_2{m}\rceil} n$. The multiplication of two elements over the finite field is equal to the addition of their exponents. Therefore, $2q^2-q$ additions are required for a unit. Note that $q=2^m$, thus, the complexity of XOR operations on GF($q$) is $m$ times that on GF($2$). Assume that mergesort is used for the comparisons among conditional probabilities, the complexity of which for nonbinary polar codes is about $\mathcal{O}(qL(\log_2qL))$ and for binary cases is about $\mathcal{O}(2L(\log_22L))$.

\begin{table}[!t]\scriptsize
\renewcommand\arraystretch{1.26}
\centering
\caption{Decoding Complexity of SCL Decoder}
\begin{tabular}{|c|c|c|c|}%{|c|c|c|c|}
	\hline
	Polar codes&Multiplications&Additions&XORs \\
	\hline
	$q$-ary&$\frac{{({q^2} + q)Ln{{\log}_2}n}}{2}$&$({q^2}-\frac{q}{2})Ln{\log _2}n$&$\frac{{m({q^2} + q)Ln{{\log }_2}n}}{2}$\\
	\hline
	Binary&$3LN{\log_2}N$&$LN{\log_2}N$&$3LN{\log_2}N$ \\
	\hline
\end{tabular}
\label{tab:2}
\end{table}

It can be seen that with the increase of the field order, the decoding complexity increased intensively, and there is a trade-off between the complexity and the latency.

\section{Combination with High-order Modulation}
In this section, the performance of two-stage polarization-based nonbinary polar codes combined with high-order modulations is considered, where two schemes are introduced: \textit{single level coded modulation scheme} and \textit{mixed multilevel coded modulation scheme}. The first one takes account of the latency, while the second one can provide lower complexity with the use of smaller fields. For the improved error-correcting performance, CRC-aided SCL method is considered for all schemes.

\subsection{Single Level Coded Modulation with Field Matched Modulation Order}
The system model under consideration is shown in Fig. \ref{fig:13}. Suppose that a two-dimensional signal constellation $\mathcal{X}$ of size $|\mathcal{X}|=2^m=q$ is used. A signal mapper $\mathcal{M}(\cdot)$ maps the polar coded symbols ${c_i, 1\leq i\leq n}$ to modulated symbols ${x}_i\in \mathcal{X}$, directly, i.e., ${x}_i=\mathcal{M}({c}_i)$.
For finite-length polar coding, the polarization effect will not so perfect, and the order of bit-channel (symbol-channel) reliabilities, i.e., code construction, plays an important role in error-correcting performance, which will be affected by the mapping function $\mathcal{M}(\cdot)$, namely, the constellation labeling.
% and the constellation labeling. In this paper, both set partition (SP) labeling and Gray labeling rules.

Let $\mathbf{\tilde c}=(c^{(1)},\ldots,c^{(m)})$ denote a binary label vector, and $X,Y$ be random variables corresponding to their lowercase versions. Since there is a one-to-one correspondence between the $X$ and $c^{(1)},\ldots,c^{(m)}$, the proposed polar coded modulation channel can be equivalently converted into $m$ equivalent subchannels in parallel by the chain rule of mutual information, i.e.,
\begin{equation}
\begin{aligned}
{I}({X};{Y})&={I}(c^{(1)},\ldots,c^{(m)};{Y})\\
&={I}(c^{(1)};{Y})+{I}(c^{(2)};{Y}|c^{(1)})+\cdots\\
&+{I}(c^{(m)};{Y}|c^{(1)},\ldots,c^{(m-1)}).
\end{aligned}
\end{equation}
Denote $\mathbf{\tilde{c}}_{1}^{i}=(c^{(1)},\ldots,c^{(i)})$ ($1\leq i \leq m$). Then, the mutual information of the equivalent subchannel $i$ can be calculated as
\begin{equation}
\begin{aligned}
&{I}(c^{(i)};{Y}|\mathbf{\tilde{c}}_{1}^{i-1})={I}(\mathbf{\tilde{c}}_{i}^{m};{Y}|\mathbf{\tilde{c}}_{1}^{i-1})-{I}(\mathbf{\tilde{c}}_{i+1}^{m};{Y}|\mathbf{\tilde{c}}_{1}^{i-1}).
\end{aligned}
\end{equation}
Thus, the corresponding capacity of each equivalent subchannel can be obtained by
\begin{equation}\small
{C_i} = \int\limits_{y \in {{Y}}}\frac{1}{|\mathcal{X}|}{\sum\limits_{\mathbf{\tilde{c}}_{1}^{m} \in {\{0,1\}^m}} {{P_{Y|X}}(y|\mathbf{\tilde{c}}_{1}^{m})\log\frac{{{P_{Y|X}}(y|\mathbf{\tilde{c}}_{1}^{i})}}{{P_{Y|X}}(y|\mathbf{\tilde{c}}_{1}^{i-1})}}}dy,
\end{equation}
where $P_{Y|X}$ denotes the channel transition probability with input $X$ and output $Y$. The different subchannel capacities ${C_i}$ can also be regarded as a kind of polarization, which implies that the bit-channel in a coded symbol with high reliability should correspond to the subchannel with high capacity to improve the code construction.

According to the binary polarization matrix $\mathbf{H}_m$ given in Appendix A, bit-channel reliabilities in each symbol are roughly in ascending orders. Therefore, in this paper, we assume that the labeling with ascending subchannel capacities, i.e., $C_1\leq \cdots\leq C_m$, is used, and each coded symbol $c_i$ ($1\leq i\leq n$) is labeled by $\mathbf{b}(c_i)$.

%\subsubsection{Numerical Results}
\textbf{Example 7:}
Consider two-stage polarization-based $16$-ary polar codes with $16$-QAM over the AWGN channel, where the equivalent code length $N=2048$, and $R=1/2$. Assume that Gray labeling is applied.

For comparison, the performance of comparable binary polar codes, MR-based nonbinary polar codes with \textit{symbol-level} computation are also shown in Fig. \ref{fig:8}. Also considered for comparison are the performance of two-stage polarization-based nonbinary polar codes with the Gray labeling with ascending/descending subchannel capacities, where for labeling with ascending capacities (LAC), $C_1=C_2<C_3=C_4$, and for labeling with ascending capacities (LDC), $C_1=C_2>C_3=C_4$. All nonbinary polar codes are constructed by Monte-Carlo method at $E_b/N_0=3.5$ dB. For binary polar codes, bit-interleaved coded modulation (BICM) scheme~\cite{Kai2013} is considered, where the codes are constructed by Monte-Carlo method at $E_b/N_0=4.5$ dB, and the bit-interleaver designed in~\cite{Chen2016} is applied. CRC-16 is used to all above schemes.

\begin{figure}[h]
	\centering
	\includegraphics[width=3.36in]{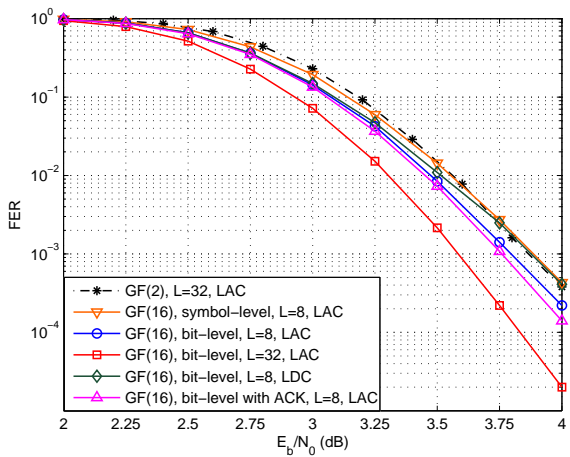} %fig17 15 20 22
	\caption{Performance comparison with 16-QAM over the AWGN channel.}
	\label{fig:8}
\end{figure}

From Fig. \ref{fig:8}, it can be seen that the two-stage polarization-based nonbinary polar codes combined with field matched modulation exhibits better performance than \textit{symbol-level} constructed codes, and the codes with LAC outperform that with LDC. Moreover, with high-order modulation, the performance can also be improved by considering active-check bits. Compared with BICM-based binary polar codes, about $0.32$ dB coding gain can be obtained by our proposed polar codes with the same list size at FER$=4\times10^{-4}$.

\subsection{Mixed Multilevel Coded (MLC) Modulation with Arbitrary Modulation Order}
As shown in Section V, the decoding latency can be decreased with the increase of the field order, while the decoding complexity will be increased intensively. Thus, for large modulation order, the symbol-based coding scheme is not efficient. Assume that $|\mathcal{X}|=2^r$, where $r$ is a positive integer greater than unity. Considering the trade-off between complexity and latency, a mixed polar coded MLC scheme is provided, where each level uses a $q_i=2^{m_i}$-ary polar code as the component code $\mathcal{C}_i$. For simplicity, in this paper, we assume that the code length of all the component codes is equal to $n$, then we have $1\leq i\leq I$ with $r=\sum_{i=1}^{I}m_i$.

The system model is shown in Fig. \ref{fig:19}. A $K$-bit information sequence $\mathbf{m}$ is first fed into a $t$-bit CRC outer encoder, resulting in a binary sequence $\mathbf{d}$ of length $K+t$. Then, the sequence is partitioned into $I$ subsequences, i.e., $\mathbf{d}=(\mathbf{d}^{(1)},\ldots,\mathbf{d}^{(I)})$. Let $K_i$ denote the length of sequence $\mathbf{d}^{(i)}$ with $\sum_{i=1}^{I}K_i=K+t$. Similar to binary polar MLC schemes, the value of $K_i$ is decided by the code construction~\cite{Seidl2013}, i.e., according to the equivalent bit-channel reliabilities via Monte-Carlo method.
Each sequence $\mathbf{d}^{(i)}$ is encoded by an individual polar encoder, producing a codeword $\mathbf{c}^{(i)}=({c}_1^{(i)},\ldots,{c}_{n}^{(i)})$ ($1\leq i\leq I$) of the component code $\mathcal{C}_i$.
Note that the total code rate is $R_t=\frac{K+t}{\sum_{i=1}^{I}m_in}$.

The coded symbols $\mathbf{c}_j=\{c_j^{(i)}, 1\leq i\leq I\}$ at time $j$ ($1\leq j \leq n$) are mapped into the signal constellation $\mathcal{X}$. The process can be regarded as $I$ parallel symbol-based coding schemes with the signal set $\mathcal{X}_i$ for $i$-th coding level, in which $|\mathcal{X}_i|=2^{m_i}$ and $\mathcal{X}=\mathcal{X}_1\times\ldots\times\mathcal{X}_I$.
At the receiver, the component codes $\mathcal{C}_i$ are successively decoded by the SCL decoder. The CRC decoder works after the decoding of all coding levels, and the decoder outputs the estimated information sequence given by the decoding path with the largest probability among the paths which can pass the CRC.
\begin{figure*}[t]
	\centering
	\includegraphics[width=5.0in]{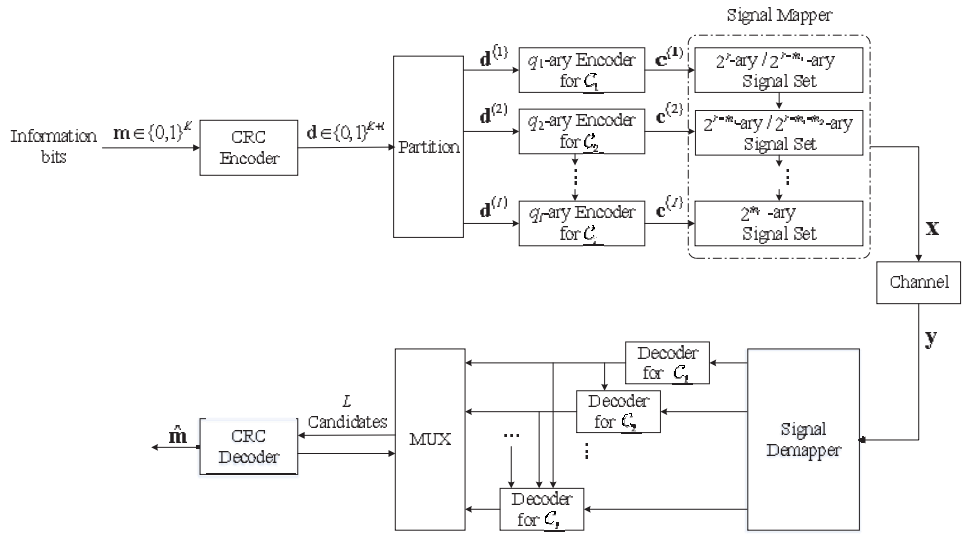}
	\caption{System model of mixed polar coded MLC scheme.}
	\label{fig:19}
\end{figure*}

\begin{figure}[!t]
	\centering
	\includegraphics[width=3.36in]{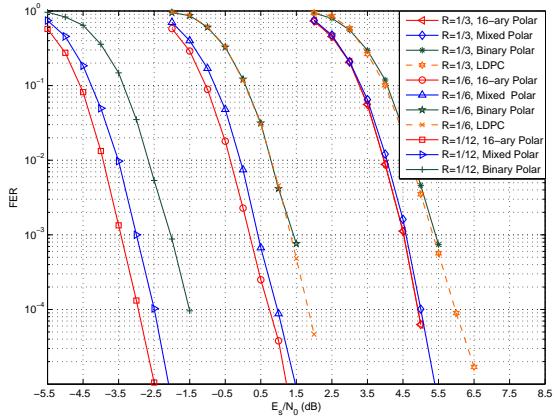}%23
	\caption{Performance comparison 16-QAM over the AWGN channel.}
	\label{fig:20}
\end{figure}

\textbf{Example 8:}
Consider the performance comparison with $16$-QAM in URLLC, where three low code rates $1/3$, $1/6$, and $1/12$ with $K=171$ are considered. The performance of both nonbinary polar codes with symbol-based coding scheme and mixed multilevel coding scheme are given, where ACK is applied to all schemes, and for the mixed MLC scheme, $I=2$, $q_1=2$, and $q_2=8$. Binary polar coded BICM scheme is also considered in Fig. \ref{fig:20}. All polar codes are constructed by Monte-Carlo method at $E_b/N_0=0.0$ dB, and the CRC-8 with $L=8$ are used to all schemes. For comparison, the performance of LDPC codes with $R=1/3$ and $R=1/6$ are also exhibited in Fig. \ref{fig:20}, where SPA with $50$ iterations is employed. The set partition (SP) labeling is applied to all mixed MLC schemes, and the Gray labeling is applied to the remaining schemes.

In Fig. \ref{fig:20}, although the mixed MLC scheme shows an inferior performance with the SLC scheme, it performs better than both LDPC and binary polar coded BICM scheme, where about $1.0$ dB coding gain can be obtained by the mixed MLC scheme at FER=$1\times 10^{-4}$. In addition, with the increase of the code rate, the gap between the two nonbinary coded schemes is decreased.

\section{Conclusion}
We have presented a new class of nonbinary polar codes constructed with two-stage polarization for URLLC, where the \textit{symbol-level} polarization is based on a $q$-ary kernel, which is a variation of Ar{\i}kan's kernel, and the \textit{bit-level} polarization is considered by using a linear transformation with a designed binary matrix. Simulation results show that the proposed nonbinary polar codes perform better than binary polar codes with BPSK and high-order modulations over the AWGN channel, and exhibit low decoding latency.
As a future work, low complexity decoding algorithms should be investigated for the proposed nonbinary polar codes.

\section*{Acknowledgment}
We wish to thank Mr. Huaan Li and Mr. Zhen Liu for providing the performance of LDPC codes.

\appendix
\subsection{Binary Polarization Matrices for Linear Transformation}
\begin{table}[!h]\footnotesize
	\centering
	\renewcommand{\multirowsetup}{\centering}
	\renewcommand\arraystretch{1.16}
	\caption{Binary Polarization Matrices}
	\begin{tabular}{|c|c|c|}
		\hline		
	     Field Size&Field Order&$\mathbf{H}_m$\\
		\hline
		$q=4$&$m=2$&$\mathbf{H}_2=\left[ {\begin{array}{*{20}{c}}
			1&0\\
			1&1
			\end{array}}\right]$\\
		\hline
		$q=8$&$m=3$&$\mathbf{H}_3=\left[ {\begin{array}{*{20}{c}}
			1&0&0\\
			1&1&0\\
			0&1&1
			\end{array}}\right]$\\
		\hline
		$q=16$&$m=4$&$\mathbf{H}_4=\left[ {\begin{array}{*{20}{c}}
			1&0&0&0\\
			1&0&1&0\\
			1&1&0&0\\
			1&1&1&1\\
			\end{array}}\right]$\\
		\hline
	\end{tabular}
	\label{tab:4}
\end{table}
\subsection{Proof of Proposition 1}

Define an ensemble of random variables $(U_{11},U_{12},U_{21},U_{22},U_1,U_2,V_{11},V_{12},V_{21},V_{22},V_1,V_2,C_{11},C_{12},$
$C_{21},C_{22},C_1,C_2,\tilde Y_{11},\tilde Y_{12},\tilde Y_{21},\tilde Y_{22},Y_1,Y_2)$ so that $(U_{11},U_{12},U_{21},U_{22})$ is uniformly distributed over $\mathbb{F}_2^4$,
$\mathbf{b}(V_1)=(V_{11},V_{12})=\mathbf{b}(\mathcal{T}_2(U_1))=(U_{11}\oplus U_{12}, U_{12})$, $\mathbf{b}(V_2)=(V_{21},V_{22})=\mathbf{b}(\mathcal{T}_2(U_2))=(U_{21}\oplus U_{22}, U_{22})$, $\mathbf{b}(C_1)=(C_{11},C_{12})=\mathbf{b}(V_1+ \alpha V_2)$, $\mathbf{b}(C_2)=(C_{21},C_{22})=\mathbf{b}(V_2)$ and $(Y_1,Y_2)=(f(\tilde Y_{11},\tilde Y_{12}),f(\tilde Y_{21},\tilde Y_{22}))$.

From the fact that $\mathcal{\tilde Y}^m\rightarrow \mathcal{Y}$ is invertible, we have
\begin{equation}
\begin{split}
I(W'_{1})&=I(U_{11};f(Y_{11},Y_{12}),f(Y_{21},Y_{22}))\\
&=I(U_{11};Y_1,Y_2)\\
I(W'_{2})&=I(U_{12};f(Y_{11},Y_{12}),f(Y_{21},Y_{22}),U_{11})\\
&=I(U_{12};Y_1,Y_2,U_{11})\\
I(W''_{1})&=I(U_{21};f(Y_{11},Y_{12}),f(Y_{21},Y_{22}),U_{11},U_{12})\\
&=I(U_{21};Y_1,Y_2,V_1)\\
I(W''_{2})&=I(U_{22};f(Y_{11},Y_{12}),f(Y_{21},Y_{22}),U_{11},U_{12},U_{21})\\
&=I(U_{22};Y_1,Y_2,V_1,U_{21})\\
\end{split}
\end{equation}
Since $U_{11},U_{12},U_{21},U_{22}$ are independent, we have\\
\begin{equation*}
\begin{split}
&{I(U_{12};Y_1,Y_2,U_{11})=I(U_{12};Y_1,Y_2|U_{11})}\\ &{I(U_{21};Y_1,Y_2,V_1)=I(U_{21};Y_1,Y_2|V_1)}\\ &{I(U_{22};Y_1,Y_2,V_1,U_{21})=I(U_{22};Y_1,Y_2|V_1,U_{21})}.
\end{split}
\end{equation*}

So, by the chain rule, we get
\begin{equation}
\begin{split}
&{I(W'_{1})+I(W'_{2})=I(U_{11},U_{12};Y_1,Y_2)=I(W')}\\ &{I(W''_{1})+I(W''_{2})=I(U_{21},U_{22};Y_1,Y_2|V_1)=I(W'')}.\\
\end{split}
\end{equation}
and
\begin{equation}
\begin{split}
&I(W'_{1})+I(W'_{2})+I(W''_{1})+I(W''_{2})\\
&=I(U_{11},U_{12},U_{21},U_{22};Y_1,Y_2)\\
&=I(U_{11},U_{12},U_{21},U_{22};f(Y_{11},Y_{12}),f(Y_{21},Y_{22}))\\
&=I(C_{11},C_{12},C_{21},C_{22};Y_{11},Y_{12},Y_{21},Y_{22})\\
\end{split}
\end{equation}
The proof of (7) is completed by noting that $I(C_{11},C_{12},C_{21},C_{22};Y_{11},Y_{12},Y_{21},Y_{22})=\sum_{i=1}^{2}\sum_{j=1}^{2}$
$I(C_{ij};Y_{ij})=4I(\tilde W)$ and $I(W')+I(W'')=2I(W)$.
Furthermore, similar method in Appendix C in~\cite{Arikan2009} can be use to prove (8) and (9), where by noting that
\begin{equation}
\begin{split}
&I(W'_{2})=I(U_{12};Y_{12})+I(U_{12};Y_{11},Y_{21},Y_{22},U_{11}|Y_{12})\\
&= I(\tilde W)+I(U_{12};Y_{11},Y_{21},Y_{22},U_{11}|Y_{12})\\
&\geq I(\tilde W)\\
&I(W''_{2})=I(U_{22};Y_{22})+I(U_{22};Y_{11},Y_{12},Y_{21},U_{11},U_{12},U_{21}|Y_{22})\\
&= I(\tilde W)+I(U_{22};Y_{11},Y_{12},Y_{21},U_{11},U_{12},U_{21}|Y_{22})\\
&\geq I(\tilde W).
\end{split}
\end{equation}

\balance

\end{document}